\documentclass{article}
\usepackage{ijcai25}

\usepackage{times}
\usepackage{soul}
\usepackage{xcolor}
\usepackage{url}
\usepackage[hidelinks]{hyperref}
\usepackage[utf8]{inputenc}
\usepackage[small]{caption}
\usepackage{graphicx}
\graphicspath{ {./images/} }
\usepackage{natbib}
\usepackage{amsmath}
\usepackage{amsthm}
\usepackage{amssymb}
\usepackage{booktabs}
\usepackage{enumitem}
\usepackage[switch]{lineno}
\usepackage{subcaption}
\usepackage[linesnumbered,ruled]{algorithm2e}

\newtheorem{theorem}{Theorem}[section]

\newtheorem{lemma}[theorem]{Lemma}

\newtheorem{definition}{Definition}
\newcommand{\eps}{\left\lceil \frac{1}{\epsilon} \right\rceil}

\title{A Generalisation of Voter Model: Influential Nodes and Convergence Properties}

\author{
Abhiram Manohara$^1$
\and
Ahad N. Zehmakan$^2$
\affiliations
$^1$ Indian Institute of Science, Bangalore, India\\
$^2$ School of Computing, The Australian National University, Canberra, Australia
\emails
abhiramm@iisc.ac.in,
ahadn.zehmakan@anu.edu.au\\
November 2024
}

\begin{document}

\maketitle

\begin{abstract}
Consider a social network where each node (user) is blue or red, corresponding to positive or negative opinion on a topic. In the voter model, in discrete time rounds, each node picks a neighbour uniformly at random and adopts its colour. Despite its significant popularity, this model does not capture some fundamental real-world characteristics such as the difference in the strengths of connections, individuals with no initial opinion, and users who are reluctant to update. To address these issues, we introduce a generalisation of the voter model. 

We study the problem of selecting a set of seed blue nodes to maximise the expected number of blue nodes after some rounds. We prove that the problem is NP-hard and provide a polynomial time approximation algorithm with the best possible approximation guarantee. Our experiments on real-world and synthetic graph data demonstrate that the proposed algorithm outperforms other algorithms.

We also prove that the process could take an exponential number of rounds to converge. However, if we limit ourselves to strongly connected graphs, the convergence time is polynomial and the convergence period (size of the stationary configuration) is bounded by the highest common divisor of cycle lengths in the network.
\end{abstract}

\section{Introduction}

Humans constantly form and update their opinions on different topics, from minor subjects such as which movie to watch and which new caf\'{e} to try to major matters such as which political party to vote for and which company to invest in. In the process of making such decisions, we tend to rely not only on our own personal judgement and knowledge, but also that of others, especially those whose opinion we value and trust. As a result, opinion diffusion and influence propagation can affect different aspects of our lives, from economy and defence to fashion and personal affairs.

Recent years have witnessed a booming development of online social networking platforms like Facebook, WeChat, and Instagram. The enormous popularity of these platforms has led to fundamental changes in how humans share and form opinions. Social phenomena such as disagreement and polarisation that have existed in human societies for millennia, are now taking place in an online virtual world and are tightly woven into everyday life, with a substantial impact on society. As a result, there has been a growing demand for a quantitative understanding of how opinions form and diffuse because of the existence of social ties among a community’s members. Within the field of computer science, especially computational social choice, there has been a rising interest in developing and analysing mathematical models which simulate the opinion diffusion in a network of individuals, cf.~\cite{bredereck2017manipulating, faliszewski2022opinion, bredereck2021maximizing,gartner2017color,shirzadi2024stubborn}.
A very popular opinion diffusion model is the Voter Model (VM)~\cite{Has99}. Consider a social network $G$, where each node (user) is either blue or red. In each round, every node picks a neighbour at random and adopts its colour. Red and blue can, for example, represent a positive/negative opinion about a topic/product. This captures the setup where switching colours is free or inexpensive; for example, changing opinions about a controversial topic or switching from one grocery store chain to another.

\textbf{First Contribution: Model Generalisation.} Despite being simple and intuitive, VM has some fundamental shortcomings. Firstly, it assumes the underlying graph is undirected and unweighted. We allow the graph to be directed (modelling one-directional relations such as following in Instagram) and weighted (modelling the strengths of the relationships). Self loops indicate inertia against changing their opinion, thus modelling situations where switching is expensive. Secondly, unlike VM, our model permits for users which have no initial opinion (uncoloured nodes) who can then gain an opinion through interaction. Thirdly, our model considers the users who may be stubborn and don't update their opinions as a result of interaction with their neighbours.

\textbf{Second Contribution: Adoption Maximisation.} In viral marketing, one aims to convince a subset of users to adopt a positive opinion about a product (i.e., become blue) with the goal that this results in a further adoption of blue colour by many other users later in the propagation process. Motivated by this application, we study the optimisation problem of maximising the expected number of blue nodes after some rounds by selecting a fixed number of initial blue nodes. We prove that the problem cannot be approximated better than $(1-1/e)$, unless a fundamental complexity hypothesis is violated. While the proof uses a classical reduction construction from the Maximum Coverage problem, the main difficulty that we need to overcome using novel techniques, is handling the randomness involved in the process. We believe some of the tools developed can be of interest to a wider set of problems. We then provide a polynomial time algorithm with such approximation ratio. While the algorithm follows a simple greedy approach, the proof of submodularity and polynomial calculation of the objective function require novel techniques from graph theory and linear algebra. Furthermore, our experiments on different real-world graph data demonstrate that our proposed algorithm outperforms other methods.



\textbf{Third Contribution: Convergence Properties.} How long does it take for our model to converge in expectation (the convergence time) and how many states/colouring are in a converged configuration (convergence period)? (Please see Section~\ref{preliminaries:sec} for formal definitions.) Leveraging techniques from Markov chain analysis and combinatorics, we prove that both convergence time and period can be exponential in the general case. However, we provide polynomial bounds for special classes of graphs, such as strongly connected graphs, in terms of the number of nodes and the highest common divisor of cycle lengths in the graph.

\section{Preliminaries}
\label{preliminaries:sec}

\textbf{Graph Definitions.} A weighted directed \textit{graph} $G$ is an ordered triple $(V,E, \omega)$ where elements of $V$ are nodes and $E \subset V \times V$ is a set of ordered pairs of nodes called edges and $\omega:E\rightarrow \mathbb{R}^+$ is a function that assigns a positive \textit{weight} to each edge. We define $n := |V|$ and $m := |E|$. For two nodes $v_1,v_2\in V$, we say $v_1$ is an \textit{in-neighbour} of $v_2$ (and $v_2$ is an \textit{out-neighbour} of $v_1$) when $(v_1, v_2) \in E$. Furthermore, we say there is an edge \textit{from} $v_1$ \textit{to} $v_2$. Let $N_+(v_1) := \{v \in V : (v_1,v) \in E\}$ and $N_-(v_1) := \{v \in V : (v,v_1) \in E\}$ be the set of \textit{out-neighbours} and \textit{in-neighbours} of $v_1$ respectively and $d_+(v_1) := |N_+(v_1)|$ and $d_-(v_1) := |N_-(v_1)|$ be the \textit{out degree} and \textit{in degree}. Graph $G$ is said to be \textit{normalised} if $\sum_{u \in N_+(v)}{\omega((v,u)) = 1}$ for every node $\ v \in V$.

A \textit{path} is a list of \textit{distinct} nodes $P = p_0, p_1, p_2, ... ,p_l$ for $0 \leq l < n$ such that $(p_i, p_{i+1}) \in E$ for any $0 \leq i < l$. The \textit{length} of a path $l$ is the number of edges in the path. For any two nodes $S,T \subset V$, an \textit{$S$-$T$ path} (or a path from $S$ to $T$) is a path $P$ such that $S \cap P = \{p_0\}$ and $T \cap P = \{p_l\}$. A graph is \textit{strongly connected} if for any two nodes $v_1, v_2 \in V$, there exists a path from $v_1$ to $v_2$. The \textit{distance} $\mu(v_1,v_2)$ is the length of the shortest path from $v_1$ to $v_2$. The \textit{diameter}  of a graph is the largest distance between any two nodes.


\textbf{Model Definitions.} A \textit{colouring} is a function $S : V \rightarrow \{r,b,u\}$ where $r$ stands for red, $b$ for blue and $u$ for uncoloured (a node who hasn't adopted an opinion yet). We say a node is coloured if it is not uncoloured (i.e., is either blue or red). Define $B_+^{S}(v)$ to be the nodes in $N_+(v)$ which are blue in $S$. We write $B_+(v)$ when $S$ is clear from the context.

We also introduce the notion of stubbornness, where some nodes are fixed on their opinion and not willing to change it. The following definition makes it clear that we capture this notion through nodes with out degree 0.
\begin{definition}[Generalised Voter Model]
In the \textsc{Generalised Voter Model} (GVM) on a graph $G = (V,E,\omega)$ and an initial colouring $S_0$, nodes update their colour simultaneously. In each discrete-time round $t$, let the colouring be $S_t$. $S_{t+1}$ is decided node-wise as follows: each node $v$ picks an out-neighbour proportional to the weight of the edge to that node and adopts that colour if it is red or blue. If it is uncoloured or if $v$ has no out-neighbours $v$ retains its colour. Say $v$ picks some $w \in N_+(v)$; then, $S_{t+1}(v) = S_t(w)$ if $S_t(w) \neq u$, and $S_{t+1}(v) = S_t(v)$, otherwise.
\end{definition}

The nodes with no out-neighbours will retain their initial colour after each round, emulating stubbornness (or loyalty). Note that under our model, uncoloured represents nodes who are yet to be introduced to either idea. A stubborn uncoloured node represents someone who has no opinion and refuses to get one. In the setting of elections, it may be someone who dislikes politics and will abstain from voting.

If a node $v$ was not already blue in round $t$, the probability of turning blue in round $t+1$  is the probability of picking a blue neighbour. Since the choice of neighbour is determined by the weights, this comes to be $(\sum_{u \in B_+^{S_t}(v)}\omega((v,u)))/(\sum_{u \in N_+(v)}\omega((v,u)))$. If $S_t(v) = b$, we have to add the probability of picking an uncoloured neighbour since in that case $v$ remains blue. A similar argument applies to red and uncoloured case. Note that if we normalise the graph by dividing the weight of each outgoing edge for a node $v$ by $\sum_{u \in N_+(v)}\omega((v,u))$, the probabilities of picking each neighbour is not affected and hence our model is not affected. Furthermore, if $\omega'$ given by $\omega'((v,w)) := \omega((v,w))/(\sum_{u \in N_+(v)}\omega((v,u)))$ for all $(v,w) \in E$ are the normalised weights, the probability of picking a neighbour is the weight of the edge to that neighbour. For example, the aforementioned probability can be rewritten as $\sum_{u \in B_+^{S_t}(v)}\omega'((v,u))$. Going forward, we will assume that all graphs are normalised.


Our model reduces to the original VM~\cite{Has99}, when there are no stubborn nodes, the graph is unweighted and undirected, and all nodes are either blue or red (no uncoloured nodes).

Since this is a probabilistic process, at each time $t>0$,  node $v$ has a probability of being in each colour. Thus, by misusing the notation, we represent $S_t(v)$ as a vector of length 3 indicating probability of being coloured $r$, $b$, and $u$ at time $t$, when it's clear from the context. Further, assuming an ordering of the nodes, we represent $S_t$ itself as a $n \times 3$ matrix where each row is the vector for one node. For $1 \leq i \leq n$ and $c \in \{r,b,u\}$, we shall use $S_t(v_i, c)$ to refer to the probability of $v_i$ having colour $c$ at time $t$.

\textbf{Adoption Maximisation.} Now, we are ready to introduce our viral marketing problem.

\begin{definition}[Adoption Maximisation (AM) Problem]
    Let $\Omega = (G=(V,E,\omega), S)$ be a system. For $A \subset V$ define $\Omega_A$ to be $(G,S')$ where $S'$ is given by $S'(v) = b$ for $v \in A$ and $S'(v) = S(v)$ for $v \in V \setminus A$ and $F_\tau(A) = \sum_{v \in V} S_\tau^{\Omega_A}(v,b)$ is the expected number of blue nodes at time $\tau$. For a system $\Omega$, a time $\tau$ and a budget $k$ the AM problem is to find $$\underset{A \subset V, |A| \leq k}{argmax} F_\tau(A)$$
\end{definition}
\textbf{Convergence Properties.}
Our model corresponds to a Markov chain, where the states correspond to all possible $3^n$ colourings and there is an edge from one state to another if there is a non-zero transition probability. Consider the directed graph of the Markov chain with the node (state) set $\mathcal{S}$. A strongly connected component is a maximal node set such the subgraph induced by the node set is strongly connected. The node set $\mathcal{S}$ can be partitioned to strongly connected components. On contracting these components to single nodes, we get the component graph $C_G$ of the graph. $C_G$ then has a topological ordering. The absorbing strongly connected components are the components that correspond to nodes with out degree 0 in $C_G$ and we shall call them the leaves of the graph. The process is said to converge when it enters one of the leaves of $C_G$. The number of rounds the process needs to enter a leaf is the \textbf{\textit{convergence time}} and the size of the leaf component (number of states) is the \textbf{\textit{period}} of convergence.

\section{Related Work}
\label{previous-sec}

\textbf{Models.} Numerous opinion diffusion models have been developed to understand how members of a community form and update their opinions through social interactions with their peers, cf.~\cite{noorazar2020recent, brill2016pairwise, wilczynski2019poll,gartner2017color,li2023graph}. As mentioned, our main focus is on the generalisation of Voter Model, which was introduced originally in~\cite{Has99}, and has been studied extensively afterwards, cf.~\cite{petsinis2023maximizing,gauy2025votermodelmeetsrumour}.

We first give a short overview of some of the most popular opinion diffusion models. The \textit{Independent Cascade} (IC) model, popularised by the seminal work of~\cite{kempe2003maximizing}, has obtained substantial attention to simulate viral marketing, cf.~\cite{li2018influence,liu2023fast}. In this model, initially each node is uncoloured (inactive), except a set of seed nodes which are coloured (active). Once a node is coloured, it gets one chance to colour each of its out-neighbours. Different extensions of the IC model have been introduced, cf.~\cite{lin2015analyzing, myers2012clash}. The IC model aims to simulate the spread of influence or the adoption of novel technology (with no competitor).
In the \textit{Threshold model}, each node $v$ has a threshold $\tau(v)$. From a starting state, where each node is either coloured or uncoloured, an uncoloured node becomes coloured once $\tau(v)$ fraction of its out-neighbours are coloured. Different versions of the threshold model have been studied, cf.~\cite{n2020rumor,gartner2020threshold,out2021majority}. In the \textit{Majority model}. cf.~\cite{chistikov2020convergence, zehmakan2024majority}, in every round each node updates its colour to the most frequent colour in its out-neighbourhood. Unlike the IC or Threshold model, here a node can switch back and forth between red and blue (similar to ours).

\textbf{Adoption Maximisation.} For various models, the problem of finding a seed set of size $k$ which maximises the expected number of nodes coloured with a certain colour after some rounds have been studied extensively.
This problem is proven to be NP-hard in most scenarios, and thus the previous works have resorted to approximation algorithms for general case (see~\cite{lu2015competition, zehmakan2019spread}) or exact algorithms for special cases (see~\cite{bharathi2007competitive,zehmakan2023random}). For example, for the Threshold model, the problem cannot be approximated within the ratio of $O(2^{\log^{1-\epsilon}n})$, for any constant $\epsilon>0$, unless $NP\subseteq DTIME(n^{polylog(n)})$~\cite{chen2009approximability}. However, the problem is traceable for trees~\cite{centeno2011irreversible} and there is a $(1-1/e)$-approximation algorithm for the Linear Threshold (LT) model, where the threshold $\tau(v)$ is chosen uniformly at random in $[0,1]$, cf.~\cite{kempe2003maximizing}.

For the Voter Model, it was proven that the final fraction of blue nodes is equal to the summation of the degree of all initially blue nodes divided by the summation of all degrees~\cite{Has99}. Thus, a simple algorithm which picks the nodes with the highest degree solves the problem in polynomial time. However, as we will prove, the problem is computationally much harder in our more general setup. The problem also has been proven to be NP-hard, by~\cite{even2007note}, when each node has a cost and the goal is to maximise the expected number of blue nodes in round $\tau$ of the Voter Model for a given cost.

\textbf{Convergence Properties.}
The convergence time is arguably one of the most well-studied characteristic of dynamic processes, cf.~\cite{auletta2018reasoning, auletta2019consensus}. For the Majority model on undirected graphs, it is proven by~\cite{poljak1986pre} that the process converges in $\mathcal{O}(n^2)$ rounds (which is tight up to some poly-logarithmic factor~\cite{frischknecht2013convergence}). The convergence properties have also been studied for directed acyclic graphs~\cite{chistikov2020convergence}, expander graphs~\cite{zehmakan2020opinion}, and when the updating rule is biased~\cite{lesfari2022biased,zehmakan2021majority}. It has recently also been studied for general directed graphs under the voter model [\cite{gauy2025votermodelmeetsrumour}].
For the Voter Model, an upper bound of $\mathcal{O}(n^3\log n)$ has been proven in~\cite{Has99} using reversible Markov chain argument.
~\cite{abdullah2015global} considered a model similar to the Voter Model with two alternatives, where in each round every node picks $k$ of its neighbours at random and adopts the majority colour among them. They proved that starting from a random initial colouring, the process converges in $\mathcal{O}(\log_k\log_k n)$ rounds in expectation.

\section{Maximum Adoption Problem}
\label{max-sec}

\subsection{Innaproximability}
\begin{theorem}
\label{hardness}
    There is no polynomial time $(1 - \frac{1}{e} +\epsilon)$-approximation algorithm (for any constant $\epsilon > 0$) for the Adoption Maximisation problem, unless $NP \subseteq DTIME(n^{O(\log \log n)})$.
\end{theorem}
\textit{Proof Sketch.} (The detailed proof is presented in the Appendix) We will prove this by reducing any instance of the Maximum Coverage problem (cf.~\cite{khuller1999budgeted}) to an instance of the Adoption Maximisation Problem. Consider an instance of the Maximum Coverage problem. Let $\mathcal{O} = \{O_1,O_2, \cdots O_m\}$ be the set and $\mathcal{S} = \{S_1, S_2, \cdots S_l\}$ be the collection of subsets of $\mathcal{O}$. We need to find $\mathcal{A} \subset \mathcal{S}$ of size $k$ that maximises $\bigcup_{S \in \mathcal{A}} S$. If $k=0$, the answer is $0$. If $k \geq l$ we just pick all the subsets and if $\exists \ O \notin \bigcup_{S \in \mathcal{S}} S$, it can't be covered, so we can safely ignore it. Thus, assume $0 < k < l$ (this implies that $l \geq 2$ since $k$ is an integer) and $\bigcup_{S \in \mathcal{S}} S = \mathcal{O}$. If $k\geq m$, there is a solution of size $m$ since for each object, we can pick a subset which includes it. So assume $k<m$. A similar argument yields that the optimal value is at least $k$.


Now consider node sets $V_{\mathcal{O}} := \{o_1, \cdots o_m\}$ and $V_{\mathcal{S}} := \{s_1, \cdots s_l\}$. Add edges $(o_i,s_j)$ if and only if $O_i \in S_j$. Additionally, for each $1 \leq j \leq m$, we introduce $d=max(\eps, m)$ more nodes $o_j^1, \cdots o_j^d$, and edges $(o_j^1,o_j), \cdots (o_j^d,o_j)$. All edges have weight 1 (before normalisation) and all nodes are initially uncoloured. We also set $\tau = ld+1$. This completes the construction of an instance of our problem from the Maximum Coverage problem. (We also observe that $(1 - \frac{1}{e} +\epsilon)$ must be at most $1$. Thus, $\epsilon < \frac{1}{e}$ and $\eps \geq 2$.)

\begin{lemma}
\label{ineq}
For the above choices of parameters, the following inequalities hold:
\begin{enumerate}[label=\Alph*.]
    \item $1-\left(1-\frac{1}{l}\right)^\tau - d\left(1-\frac{1}{l}\right)^{\tau-1} > 0$
    \item $\frac{\epsilon}{2}>\left(1-\frac{1}{l}\right)^{\tau-1}\left(1-\frac{1}{e}+\epsilon\right)$
    \item $d+1 > \frac{\frac{1}{e} - \left(1-\frac{1}{l}\right)^{\tau-1}}{\epsilon - \left(1-\frac{1}{l}\right)^{\tau-1}\left(1-\frac{1}{e}+\epsilon\right)}$
\end{enumerate}
\end{lemma}

\begin{lemma}
\label{a-vs-lemma}
Let $A'$ be a solution to the instance of our problem such that $A'\not\subset V_{\mathcal{S}}$, then there is a strictly better solution $A$ such that $A\subset V_{\mathcal{S}}$.
\end{lemma}

The above lemma implies that any solution $A'\not\subset V_{\mathcal{S}}$ can be improved (in fact, in linear time, as discussed in the proof) to a strictly better solution which has only nodes from $V_\mathcal{S}$. Particularly, the optimal solution picks only nodes from $V_\mathcal{S}$. In the following, we use this lemma to focus on the solutions where the seed set $A$ is a subset of $V_{\mathcal{S}}$.

Let $am$ be the value of a solution $A\subset V_{\mathcal{S}}$ of size $k$ in our problem and $mc$ be the corresponding value for the Maximum Coverage problem, where set $S_j$ is picked if and only if $s_j\in A$. We establish a connection between $am$ and $mc$. Firstly, it's trivial that $am \leq k + mc + d \ mc$ since by our construction at most $k + mc + d \ mc$ nodes are made blue by $A$, regardless of $\tau$. Secondly, exactly $mc$ nodes $v$ in $V_\mathcal{O}$ satisfy $|N_+(v) \cap A| \neq 0$. Then, $k$ nodes are blue with probability $1$, $mc$ nodes are blue with probability at least $1-\left(1-\frac{1}{l}\right)^\tau$ and $d \ mc$ nodes are blue with probability at least $1-\left(1-\frac{1}{l}\right)^{\tau-1}$. By linearity, we have $k + mc \left(1-\left(1-\frac{1}{l}\right)^\tau \right) + d \ mc \left(1-\left(1-\frac{1}{l}\right)^{\tau-1}\right) \leq am$. Using the two aforementioned inequalities and some simple calculations, we can establish a connection between $am$ and $mc$, as presented in the lemma below.


\begin{lemma}
\label{lemma4}

\begin{equation*}
        \left(1-\left(1-\frac{1}{l}\right)^{\tau-1}\right)\frac{am-k}{D} < mc \leq \frac{am-k}{D}
\end{equation*}
where $D := 1-\left(1-\frac{1}{l}\right)^\tau+d\left(1-\left(1-\frac{1}{l}\right)^{\tau-1}\right)$.
\end{lemma}

We'll need one last result before we prove the theorem.
\begin{lemma}
\label{noopt}
Suppose $A\subset V_{\mathcal{S}}$ gives the optimal solution for the Adoption Maximisation problem, then the corresponding set $\{S_j: s_j\in A\}$ is an optimal solution for the Maximum Coverage problem.

\end{lemma}
Consider a polynomial time algorithm $ALG_{AM}$ with approximation ratio of $\left(1 - \frac{1}{e} + \epsilon \right)$ for the Adoption Maximisation problem and some $\epsilon>0$. Then, we design an algorithm $ALG_{MC}$ which transforms a given instance of the Maximum Coverage problem to an instance of the Adoption Maximization following our polynomial time construction, runs $ALG_{AM}$ on this construction, and translates the outcome to a solution of the Maximum Coverage as explained above. We claim that $ALG_{MC}$ has an approximation ratio better than $1-\frac{1}{e}$, which we know to not be possible unless $NP \subseteq DTIME(n^{O(\log \log n)})$, cf.~\cite{khuller1999budgeted}. Thus, it remains to prove this claim.

Let $\overline{AM}$ be the solution produced by $ALG_{AM}$ and $\overline{MC}$ be the solution produced by $ALG_{MC}$ as described above. Also, let $AM$ and $MC$ denote the optimal solutions. We then can use Lemma~\ref{lemma4} to connect the values of $\overline{AM}$ and $\overline{MC}$, and additionally Lemmas~\ref{a-vs-lemma} and~\ref{noopt} to establish the connection between $AM$ and $MC$. Leveraging these two connections and some small calculations (please see Appendix \ref{cal-hardness-appendix} for omitted calculations in the rest of the proof), we get

\begin{equation}
\label{eq-star}
\begin{split}
\overline{MC} &> \left(1-\frac{1}{e}\right)MC+\epsilon MC - \left(1-\frac{1}{l}\right)^{\tau-1} \left(1 - \frac{1}{e} + \epsilon \right)MC \\ & + \frac{\left(\epsilon-\frac{1}{e}\right)k}{D}\left(1-\left(1-\frac{1}{l}\right)^{\tau-1}\right) \hfill
\end{split}
\end{equation}
Furthermore, by using the inequalities in Lemma~\ref{ineq} and some calculations, we get 

\begin{equation}
\label{eq-cal-2}
\epsilon - \left(1-\frac{1}{l}\right)^{\tau-1}\left(1-\frac{1}{e}+\epsilon\right) > \frac{\left(\frac{1}{e} - \epsilon\right)\left(1-\left(1-\frac{1}{l}\right)^{\tau-1}\right)}{D}   
\end{equation}

Finally, by applying Equation~\eqref{eq-cal-2} and the fact that the LHS of Equation~\eqref{eq-cal-2} is positive according to Lemma~\ref{ineq} (B), we can show that the RHS of Equation~\eqref{eq-star} is larger than $(1-1/e)MC$. Thus, we can conclude that $\overline{MC} > \left(1-\frac{1}{e}\right)MC$, that is, the polynomial time algorithm $ALG_{MC}$ has an approximation ratio better than $1-1/e$. $\hfill \square$

\subsection{Greedy Algorithm}
\label{greedysection}

In this section, we provide a greedy algorithm whose approximation ratio matches the lower bound proven in the previous section and runs in polynomial time when $\tau=\textrm{poly}(n)$. Let us start by proving some properties of the objective function $F_t(A)$. 

\begin{theorem}
\label{monsub}
    $F_t(A)$ is monotone and submodular.
\end{theorem}
\textit{Proof.} To prove this, we will introduce the notation of pick sequence, inspired by~\cite{viral}. 
\begin{definition}[Pick Sequence]
    A pick sequence of length $\tau$ is a function $PS_\tau : V \times [1,\tau] \cap \mathbb{N} \rightarrow V$ where $PS_\tau(v,t)$ is the node that node $v$ picks at round $t$. $R_\tau$ is the set of all pick sequences of length $\tau$. 
\end{definition}
Let $\Pr[PS_\tau]$ be the probability that the picks made by the nodes in the first $\tau$ rounds follow the pick sequence $PS_\tau$. For $A \subset V$ let $B^{PS_\tau}(A)$ and $F^{PS_\tau}_\tau(A) = |B^{PS_\tau}(A)|$ be the set of blue nodes and the number of blue nodes at time $\tau$ for the system $\Omega_A$ if the pick sequence $PS_\tau$ is followed. Then, $$F_\tau(A) = \sum_{PS_\tau \in R_\tau}\Pr(PS_\tau)F^{PS_\tau}_\tau(A)$$
Since monotonicity and submodularity are preserved under linear combinations, it suffices to show that for each $PS_\tau$, $F^{PS_\tau}_\tau$ is monotone and submodular.
\begin{definition}[Node Sequence]
    For a pick sequence $PS_\tau$ and each time $0 \leq t \leq \tau$, the node sequence up to $t$ is the function $\gamma_t : V \rightarrow V^{2^t}$ defined recursively by $\gamma_0(v) = v \ \forall \ v \in V$ and $\gamma_t(v) = \gamma_{t-1}(PS_\tau(v,t)) || \gamma_{t-1}(v)$ where $||$ stands for concatenation.
\end{definition}
The node sequence captures the idea of tracking the neighbouring node selected in each round and the nodes selected by the neighbouring node in previous rounds to get the colour of the node at round $t$. The following lemma formalises this idea. Its proof follow a simple inductive argument detailed in Appendix \ref{nodseq-appendix}.
\begin{lemma}
\label{nodseq}
The colour of a node at time $t$ is the time 0 colour of the first coloured node in its node sequence. A node is uncoloured if all the nodes in its node sequence were initially uncoloured.
\end{lemma}
To prove monotonicity, consider $A\subset A'\subset V$. For any node $v$, if $v \in B^{PS_\tau}(A)$, then there is some node $u$ in $\gamma_\tau(v)$ that is blue, and all nodes before it are uncoloured. Then, either $u \in A$ or $u$ is blue in $\Omega$ and all nodes before $u$ are uncoloured in $\Omega$ and not in $A$. Since $A \subset A'$, $u \in A'$ so $u$ is blue in $\Omega_{A'}$ and all nodes before $u$ are either in $A'$ in which case they are blue in $\Omega_{A'}$ or not in $A'$ in which case they are uncoloured in $\Omega_{A'}$. Thus, there exists a coloured node in $\gamma_\tau(v)$ according to $\Omega_{A'}$ and the first coloured node is blue. Thus, $v \in B^{PS_\tau}(A')$. Since $v$ was arbitrary, $B^{PS_\tau}(A) \subset B^{PS_\tau}(A')$ and in particular, $F_\tau^{PS_\tau}(A) \leq F_\tau^{PS_\tau}(A')$.

For submodularity, consider $A \subset A' \subset V$ and any node $w$. By our previous result, $B^{PS_\tau}(A) \subseteq B^{PS_\tau}(A')$, $B^{PS_\tau}(A \cup \{w\})$, and $B^{PS_\tau}(A'\cup \{w\})$. Consider a node $v \in B^{PS_\tau}(A'\cup \{w\}) \setminus B^{PS_\tau}(A)$. We claim that it is in at least one of $B^{PS_\tau}(A \cup \{w\})$ or $B^{PS_\tau}(A')$. To see this, consider $\gamma_\tau(v)$. There is some $u$ which is blue in $\Omega_{A' \cup \{w\}}$ and all nodes before it are uncoloured. $u$ is either in $A' \cup \{w\}$ or it is blue in $\Omega$. All nodes before it are not in $A' \cup \{w\}$ and are uncoloured in $\Omega$. If $u$ is blue in $\Omega$, then all nodes until $u$ are not in $A$ since $A \subset A' \cup \{w\}$ and the first uncoloured node in $\gamma_\tau(v)$ is blue, so $v \in B^{PS_\tau}(A)$ contradicting our assumption. Hence, $u \in A' \cup \{w\}$. Since $A' \cup \{w\} = A' \cup (A \cup \{w\})$, $u$ is in one of $A \cup \{w\}$ or $A'$. All nodes before $u$ are in neither since they are both subsets of $A' \cup \{w\}$. Thus, $v$ belongs to one of $B^{PS_\tau}_\tau(A\cup \{w\})$ or $B^{PS_\tau}_\tau(A')$. By our assumption, it isn't in $B^{PS_\tau}(A)$. So it's in one of $B^{PS_\tau}_\tau(A\cup \{w\}) \setminus B^{PS_\tau}(A)$ or $B^{PS_\tau}(A') \setminus B^{PS_\tau}(A)$. Because $v$ was arbitrary, $B^{PS_\tau}(A'\cup \{w\}) \setminus B^{PS_\tau}(A) \subset \left(B^{PS_\tau}(A') \setminus B^{PS_\tau}(A) \cup B^{PS_\tau}(A\cup \{w\}) \setminus B^{PS_\tau}(A)\right)$. Since by our monotonicity result, $B^{PS_\tau}(A)$ is a subset of all the other 3, we get $F^{PS_\tau}_\tau(A'\cup \{w\}) - F^{PS_\tau}_\tau(A) \leq \left(F^{PS_\tau}_\tau(A') - F^{PS_\tau}_\tau(A)\right) + \left( F^{PS_\tau}_\tau(A\cup \{w\}) - F^{PS_\tau}_\tau(A)\right)$. Rearranging, we get $F^{PS_\tau}_\tau(A'\cup \{w\}) - F^{PS_\tau}_\tau(A') \leq F^{PS_\tau}_\tau(A \cup \{w\}) - F^{PS_\tau}_\tau(A)$, thus proving submodularity. This concludes the proof of the theorem $\hfill \square$

\textbf{Computing Objective Function.} Before discussing the greedy algorithm, we need to explain how to compute $F_\tau(A)$ for a set $A$. The idea of fixing picks was good for theory, but since the number of pick sequences grows exponentially in time, it very quickly becomes impractical to compute.
We instead bring back the notion of probability vector $S_t$.
We can then calculate the probability of picking colour blue at round $t$ as $P_t(v, b)=\sum_{w \in N_+(v)}\omega(v,w)S_{t-1}(w,b)$ and analogously for red and uncoloured. This is just the dot product of the row vector of $v$ in the (normalised) adjacency matrix $H$ of the graph and the blue column of $S_{t-1}$. Thus, we can obtain the colour distribution probability for all the nodes as $P_t = H \times S_{t-1}$ where row $i$ of $P_{t}$ is the triple $P_t(v_i, b), P_t(v_i, r), P_t(v_i, u)$ in round $t$. Since the picking is independent across rounds, these probabilities are independent of $S_{t-1}(v)$ and we can find $S_t(v)$. For example, $S_t(v,b) = S_{t-1}(v,b)(1-P_t(v,r)) + (1-S_{t-1}(v,b))P_t(v,b)$. This concludes one round, and repeating this procedure $\tau$ times can give us $S_\tau$, starting from $\Omega_A$. Then we can add the blue column of $S_\tau$ to get $F_\tau(A)$. Multiplying an $n \times n$ matrix with an $n \times 3$ matrix can be done in $O(n^2)$ and calculating $S_t$ from $P_t$ can be done in $O(n)$. Thus, one step of calculations costs $O(n^2)$. Doing this $\tau$ times to calculate $S_\tau$ costs $O(n^2\tau)$. Then, we can add the blue column in $O(n)$. Thus, the overall cost of calculating $F_\tau(A)$  is $O(n^2\tau)$.

\textbf{Proposed Algorithm.} The greedy algorithm works by iteratively selecting and adding the node which gives the highest increase to the expected number of blue nodes in time $\tau$. In each iteration, the algorithm needs to compute the value of $\mathcal{F}_{\tau}(\cdot)$ if each non-selected node was added to the seed set. Since there are at most $n$ such nodes to be checked, we have $k$ iterations overall, and computing $\mathcal{F}_{\tau}(\cdot)$ takes $O(n^2\tau)$ as discussed above, we can conclude that the run time of the algorithm is $O(n^3k\tau)$, which is polynomial when $\tau=\textrm{poly}(n)$. Furthermore, as the objective function is both monotone and submodular, the greedy algorithm achieves an approximation ratio of $1-1/e$, cf.~\cite{nemhauser1978analysis}.



\subsection{Experimental Comparison} 

\textbf{Data.} 
We have used real-world social network data available on SNAP database~\cite{snapnets}, including Facebook 0 ($n=347$, $m=5,038$), Twitter ($n=475$, $m=13,289$), Facebook 414 ($n=685$, $m=3,386$), Wikipedia ($n=4,592$, $m=119,882$), Bitcoin OTC ($n=6,005$, $m=35,592$), Gnutella ($n=6,300$, $m=20,777$), and Bitcoin Alpha ($n=7,604$, $m=24,186$).

\noindent \textbf{Comparison.} We compare the performance of our proposed greedy algorithm against the following centrality-based methods, where we pick the nodes with: the highest \textbf{in degree}, highest \textbf{out degree}, highest \textbf{closeness} centrality, highest \textbf{betweenness} centrality, and highest (\textbf{pagerank}) centrality. A suffix of ``red'' to any of the strategies implies that the strategy prioritises converting the red nodes to blue, so selects the highest ranked red nodes.

We observe that our algorithm not only enjoys a theoretical guarantee (as proven in the previous section), but also outperforms other algorithms on real-world data. Please see Figure~\ref{fig:graphs} for Facebook 0, Facebook 414, and Bitcoin OTC. The results for other networks are similar and are given in Appendix \ref{appendix:exp}.

\begin{figure}[h]
    \includegraphics[width=0.45\linewidth]{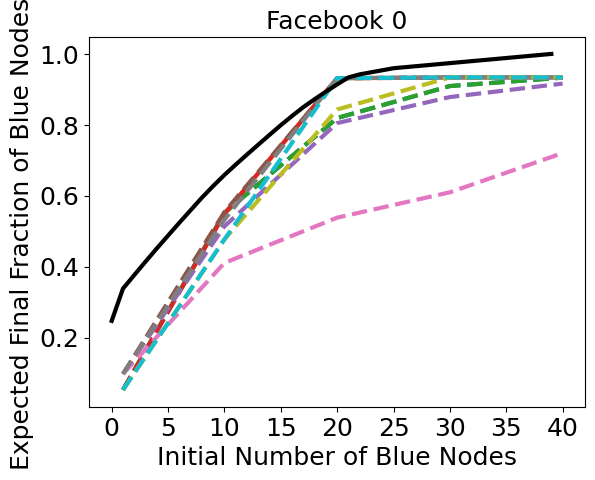}
    \includegraphics[width=0.45\linewidth]{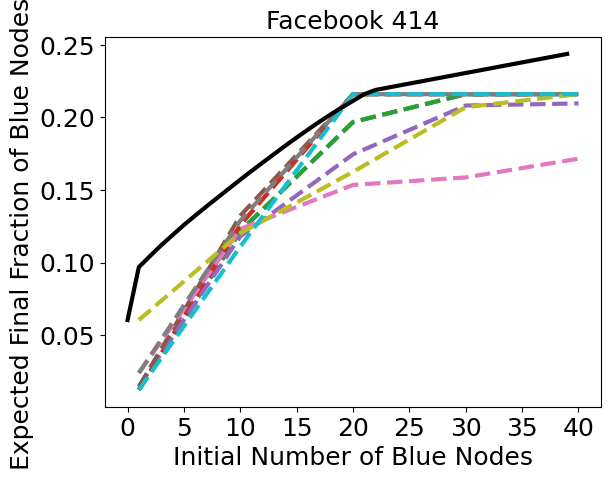} \\
    \includegraphics[width=0.45\linewidth]{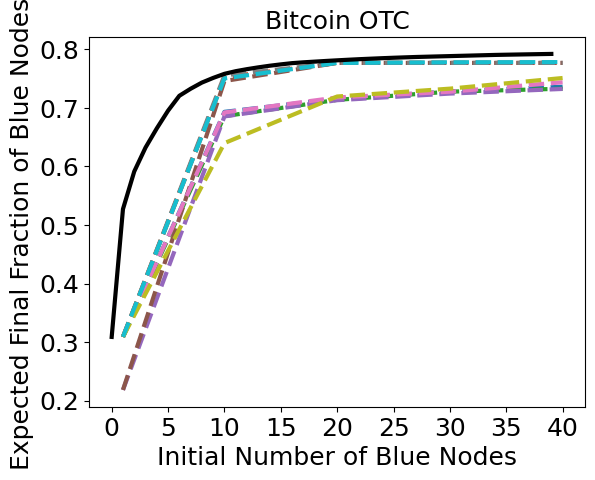} \hspace{0.15\linewidth}
    \includegraphics[width=0.3\linewidth]{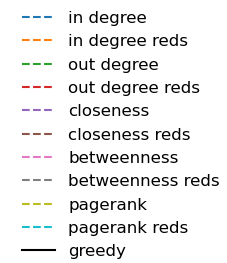} 
    \caption{Performance of Greedy algorithm against some well known centrality measures. In each graph, the final expected fraction of blue nodes is plotted against the budget. Each run has 20 red nodes and a budget varying from 1 to 40 for $\tau=20$ rounds.}
    \label{fig:graphs}
\end{figure}

    
    

\section{Convergence Properties}
\label{convergence-sec}

So far, we let the number of rounds be an input parameter $\tau$. But what kind of configurations does our model converge to, and how long does it take to reach convergence? In this section, we study the convergence period and convergence time of our process (please see Section~\ref{preliminaries:sec} for their formal definition).

\subsection{Convergence Period}

We show that if $G$ is strongly connected, the period is smaller than the HCF (highest common divisor) of the length of all cycles in $G$. However, the period can be exponential if strong connectivity condition is relaxed. 

\subsubsection{Strongly Connected Graphs}

\begin{theorem}
\label{strper}
Let the period $\gamma$ of a graph be the HCF of the lengths of all its cycles. For a strongly connected graph $G$, the period of convergence divides $\gamma$.
\end{theorem}



\begin{proof}
First, we cover the case of $\gamma=1$ in Lemma~\ref{strcon}. (the proofs of lemmas in this proof are given in Appendix \ref{appendix:convlemma}). To prove this lemma, we show that there is a path from any non-monochromatic colouring to a monochromatic one. Then, any strongly connected component that does not contain a monochromatic colouring can't be absorbing. However, the monochromatic colourings have out degree 0, thus are singleton strongly connected components that are also absorbing, thus, these are the only absorbing strongly connected components of $C_G$ (please see Section~\ref{preliminaries:sec} for the definition of $C_G$). Thus, the period is 1. The proof uses a combinatorial argument, building on the Extended Euclidean algorithm.

\begin{lemma}
\label{strcon}
A system with a strongly connected aperiodic graph ($\gamma=1$) will reach a consensus.
\end{lemma}

Let us now consider graphs of period $\gamma \ge 2$. The period classes of a graph are the equivalence classes of $V$ given by the relation $u \sim v \Leftrightarrow$ distance from $u$ to $v$ is a multiple of $\gamma$. Recall that $\mu(u,v)$ denotes the distance from $u$ to $v$

\begin{lemma}
\label{eqrel}
The relation $u \sim v \Leftrightarrow \gamma | \mu(u,v)$ is an equivalence relation in a strongly connected graph with period $\gamma$.
\end{lemma}

To prove the above lemma, we show that all three properties reflexivity, symmetry, and transitivity hold following some standard techniques. Building on this lemma and establishing a connection between the above relation $\sim$ and a newly defined relation $\sim_v$ (with respect to an arbitrary node $v$), we provide the below lemma. 

\begin{lemma}
\label{unicycle}
The graph obtained by contracting the period classes to single nodes is a single directed cycle of length $\gamma$.
\end{lemma}

For a periodic graph with period $\gamma$, consider a period class $\Gamma \subset V$. Define the graph $G_\Gamma$ with nodes $\Gamma$ and edges $(u,v)$ for $u,v \in \Gamma$ if and only if there is a $u$-$v$ path of length $\gamma$ in $G$, with weight $$\omega_\Gamma((u,v)) = \sum_{P \in P^\gamma_{(u,v)}}\prod_{e \in P}\omega(e)$$
where $P^\gamma_{(u,v)}$ is the set of all $u$-$v$ paths of length $\gamma$. The weight of an edge thus defined to be the probability of $u$ adopting $v$'s colour after $\gamma$ rounds. Then, one round in $G_\Gamma$ exactly depicts the set $\Gamma$ after $\gamma$ rounds in the process on $G$.
\begin{lemma}
\label{strap}
The graph $G_\Gamma$ defined above is strongly connected and aperiodic.
\end{lemma}

Then, by Lemma~\ref{strcon}, $\Gamma$ reaches a consensus. So each period class will be monochromatic. Then, $\gamma$ rounds later, they will have the same colour again, adopting the consensus of the previous period class at each round. So the period of convergence divides $\gamma$.
\end{proof}



\subsubsection{General Graphs}

For a strongly connected graph $G$, we provided the upper bound of $\gamma\le n$, but there is no such bound in the general case. In particular, there is a family of graphs for which the period of convergence is as big as $2^{n-2}$, which is exponential.

Consider nodes $1$ to $n$ with edges $(i,1)$ and $(i,2)$ of weight $0.5$ each for $3 \leq i \leq n$. Initially colour node 1 blue, node 2 red, and the rest uncoloured. Please see Figure~\ref{exp-conv} (left).
Nodes 1 and 2 will keep their colour unchanged forever, since they have out degree 0 (i.e., are stubborn).
On the other hand, all nodes from $3$ to $n$ will switch between blue and red independently in each round. In the Markov chain, there is an edge from every state $S$ with $S(1) = b, S(2)=r$ to every other such state and no edges to any other state. Thus, this is an absorbing strongly connected component of size $2^{n-2}$. This implies that the period of convergence in this setup is at least $2^{n-2}$

\begin{figure}[t]
\centering
\includegraphics[width=0.7\columnwidth]{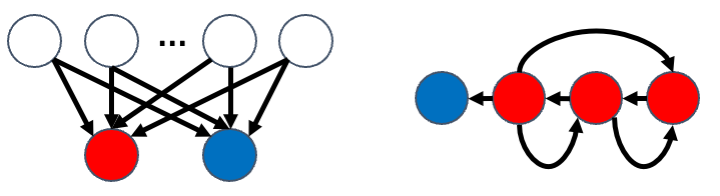}
\caption{The system constructed with (\textbf{left}) exponential period of convergence, (\textbf{right}) exponential convergence time for $n=4$.}
\label{exp-conv}
\end{figure}


\subsection{Convergence Time}
If $G$ is strongly connected, we provide a polynomial upper bound on the convergence time, but in the general case the convergence time can be exponentially large.


\subsubsection{Strongly Connected Graphs} For strongly connected graphs of period $\gamma$ (the HCF of the length of all cycles), the number of rounds required for all uncoloured nodes to be coloured is in $O(n^2 \log(n))$, following a proof from~\cite{viral}. How long does it take for the process to converge after all nodes are coloured? Similar to our proof of Theorem~\ref{strper}, we consider individual period classes $\Gamma_i$, $|\Gamma_i| = n_i$, $\sum_i n_i = n$ and their derived graphs $G_{\Gamma_i}$. The derived graph is strongly connected and aperiodic, which means that it converges in $O(n_i^3 \log(n_i))$, following the proof from~\cite{Has99}. By our construction, each round in this graph is $\gamma$ rounds in the original graph. The period class $\Gamma_i$ is expected to converge in $O(\gamma n_i^3 \log(n_i))$. There are $\gamma$ such classes, which converge independently. The system is said to converge when all classes have converged. So we need the maximum of all their convergence times. Since the convergence time is a positive random variable, we can upper bound the maximum convergence time by the sum of all convergence times. Then, by linearity of expectation, the expected convergence time is $O(\sum_i \gamma n_i^3 \log(n_i))$. Since $f(x) = x^3 \log(x)$ is convex, $\sum_if(n_i) \leq f(\sum_i n_i)$ for positive $n_i$. Also, remember $\sum_i n_i = n$. So the expected stopping time is $O(\gamma n^3 \log(n))$. 


\subsubsection{General Graphs}

Consider the graph $G$ on $n$ nodes $v_1 \cdots v_n$ with edges $(v_{i+1}, v_i)$ for every $1 \leq i \leq n-1$ and $(v_i, v_j)$ for $2 \leq i < j \leq n$ where the first node is stubborn and blue, and all the other nodes are initially red. All edges are of weight 1 before normalisation. Please see Figure~\ref{exp-conv} (right) for an example.

It is straightforward to observe that this process eventually converges to a fully blue colouring, where all nodes are blue. Thus, the convergence time is the number of rounds the process needs to reach such a colouring. We establish a connection between our process and so-called Gambler's ruin with a soft hearted adversary process~\cite{lehrer}, which then allows us to bound the convergence time.

We observe that the first time a node turns blue has to be by picking the previous node when it was blue, since no node after it can be blue without it having been blue. Furthermore, if every node after some $v$ turns red at any point, then $v$ can again only become blue by picking the previous node after that node turns blue. Let us consider the sequence $X(t) = \max_{S_t(v_i) = b} i$. We know that for the whole graph to turn blue at round $t$, we must have $X(t) = n$. And for the $n$-th node to turn blue, we need the $\frac{n}{2}$-th node to turn blue. Thus, a lower bound on $t$ for which $X(t)=n/2$ gives us a lower bound on the convergence time.

By monotonicity, the expected time from any colouring $S_t$ is at least the time from the colouring where the first $X(t)$ nodes are blue. We use this to design a simpler process. Consider a Gambler's ruin set-up with $n$ nodes where the state at time $t$ is $X(t)$. Thus, from a state $i$, it moves ahead when $v_{i+1}$ picks $v_i$. That is, with probability $\frac{1}{n-i}$. For $i \leq \frac{n}{2}$, this is $< \frac{2}{n}$. It stays in the same state when $v_{i+1}$ doesn't pick $v_i$, but $v_i$ picks $v_{i-1}$. Since the events are independent, this happens with probability $\frac{n-i-1}{n-i \cdot \frac{1}{n-i+1}}$. Again, for $i \leq \frac{n}{2}$, this is less than $\frac{4(n-2)}{n^2}$. Otherwise, state decreases. Again, we bound it by saying the state only decreases by 1, while the model can actually have all nodes pick the last node and go back to the initial state. In fact, let us be generous and say that if we didn't decrease, we will increase. Then, for the first $\frac{n}{2}$ rounds, we have that the probability of moving ahead, $p \leq \frac{2}{n} + \frac{4(n-2)}{n^2} = \frac{6n-8}{n^2} <\frac{6}{n}$. Further, let us consider $p = \frac{6}{n}$ and the probability of moving backwards, $q = \frac{n-6}{n}$. 

Using the fact that we are interested in a lower bound, we could introduce the above, much simpler process, which is at least as fast as the original process. Thus, any lower bound on this process applies to our original convergence time. This simplified process corresponds to the Gambler's ruin studied in~\cite{lehrer} whose analysis gives us the bound of $\Omega(n^n+n^2)$, which is growing exponentially.

\section{Conclusion}

We studied a generalisation of the popular Voter Model, by extension to directed weighted graphs and the introduction of stubborn agents and agents with no initial opinion. We proved the Adoption Maximisation problem is computationally hard. However, we provided a polynomial time algorithm which not only has the best possible theoretical approximation guarantee but also outperforms other algorithms on real-world data. Furthermore, we gave bounds on the period of the process in terms of graph parameters such as number of nodes, and the highest common divisor of cycle lengths.

While our proposed algorithm is polynomial, it can't handle massive graphs appearing in the real-world. Can we design faster algorithms without significant sacrifice on accuracy? Furthermore, we proved that both convergence period and convergence time can be exponential in general graphs (with tighter bounds for special graphs). What about graphs appearing in the real world? We have done some foundational studies, reported in Appendix \ref{appendix:convergence-exp}, but a deeper study of this question is left for the future work.
Finally, we have checked that our hardness and algorithmic findings, very easily extend to more setups such as (1) where the users in the selected seed set are loyal, that is, keep their colour unchanged, (2) when the goal is to optimise the average number of adoptions over the whole process rather than a fixed time $\tau$. However, the convergence properties fall short in covering a wider collection of setups. Thus, it would be interesting to study convergence properties under various scenarios, especially for different agent types.



\bibliographystyle{named}
\bibliography{ijcai25}

\begin{thebibliography}{}

\bibitem[\protect\citeauthoryear{Abdullah and Draief}{2015}]{abdullah2015global}
Mohammed~Amin Abdullah and Moez Draief.
\newblock Global majority consensus by local majority polling on graphs of a given degree sequence.
\newblock {\em Discrete Applied Mathematics}, 180:1--10, 2015.

\bibitem[\protect\citeauthoryear{Aldecoa \bgroup \em et al.\egroup }{2015}]{hgg}
Rodrigo Aldecoa, Chiara Orsini, and Dmitri Krioukov.
\newblock Hyperbolic graph generator.
\newblock {\em Computer Physics Communications}, 196:492--496, 2015.

\bibitem[\protect\citeauthoryear{Auletta \bgroup \em et al.\egroup }{2018}]{auletta2018reasoning}
Vincenzo Auletta, Diodato Ferraioli, and Gianluigi Greco.
\newblock Reasoning about consensus when opinions diffuse through majority dynamics.
\newblock In {\em IJCAI}, pages 49--55, 2018.

\bibitem[\protect\citeauthoryear{Auletta \bgroup \em et al.\egroup }{2019}]{auletta2019consensus}
Vincenzo Auletta, Angelo Fanelli, and Diodato Ferraioli.
\newblock Consensus in opinion formation processes in fully evolving environments.
\newblock In {\em Proceedings of the AAAI Conference on Artificial Intelligence}, volume~33, pages 6022--6029, 2019.

\bibitem[\protect\citeauthoryear{Bharathi \bgroup \em et al.\egroup }{2007}]{bharathi2007competitive}
Shishir Bharathi, David Kempe, and Mahyar Salek.
\newblock Competitive influence maximization in social networks.
\newblock In {\em Internet and Network Economics: Third International Workshop}, pages 306--311. Springer, 2007.

\bibitem[\protect\citeauthoryear{Bredereck and Elkind}{2017}]{bredereck2017manipulating}
Robert Bredereck and Edith Elkind.
\newblock Manipulating opinion diffusion in social networks.
\newblock In {\em IJCAI International Joint Conference on Artificial Intelligence}. International Joint Conferences on Artificial Intelligence, 2017.

\bibitem[\protect\citeauthoryear{Bredereck \bgroup \em et al.\egroup }{2021}]{bredereck2021maximizing}
Robert Bredereck, Lilian Jacobs, and Leon Kellerhals.
\newblock Maximizing the spread of an opinion in few steps: opinion diffusion in non-binary networks.
\newblock In {\em Proceedings of the Twenty-Ninth International Conference on International Joint Conferences on Artificial Intelligence}, pages 1622--1628, 2021.

\bibitem[\protect\citeauthoryear{Brill \bgroup \em et al.\egroup }{2016}]{brill2016pairwise}
Markus Brill, Edith Elkind, Ulle Endriss, and Umberto Grandi.
\newblock Pairwise diffusion of preference rankings in social networks.
\newblock In {\em International Joint Conference on Artificial Intelligence (IJCAI 2016)}, pages 130--136, 2016.

\bibitem[\protect\citeauthoryear{Centeno \bgroup \em et al.\egroup }{2011}]{centeno2011irreversible}
Carmen~C Centeno, Mitre~C Dourado, Lucia~Draque Penso, Dieter Rautenbach, and Jayme~L Szwarcfiter.
\newblock Irreversible conversion of graphs.
\newblock {\em Theoretical Computer Science}, 412(29):3693--3700, 2011.

\bibitem[\protect\citeauthoryear{Chen}{2009}]{chen2009approximability}
Ning Chen.
\newblock On the approximability of influence in social networks.
\newblock {\em SIAM Journal on Discrete Mathematics}, 23(3):1400--1415, 2009.

\bibitem[\protect\citeauthoryear{Chistikov \bgroup \em et al.\egroup }{2020}]{chistikov2020convergence}
Dmitry Chistikov, Grzegorz Lisowski, Mike Paterson, and Paolo Turrini.
\newblock Convergence of opinion diffusion is pspace-complete.
\newblock In {\em Proceedings of the AAAI Conference on Artificial Intelligence}, volume~34, pages 7103--7110, 2020.

\bibitem[\protect\citeauthoryear{Even-Dar and Shapira}{2007}]{even2007note}
Eyal Even-Dar and Asaf Shapira.
\newblock A note on maximizing the spread of influence in social networks.
\newblock In {\em Internet and Network Economics: Third International Workshop}, pages 281--286. Springer, 2007.

\bibitem[\protect\citeauthoryear{Fagen and Lehrer}{1958}]{lehrer}
R.~E. Fagen and T.~A. Lehrer.
\newblock Random walks with restraining barrier as applied to the biased binary counter.
\newblock {\em Journal of the Society for Industrial and Applied Mathematics}, 6(1):1--14, 1958.

\bibitem[\protect\citeauthoryear{Faliszewski \bgroup \em et al.\egroup }{2022}]{faliszewski2022opinion}
Piotr Faliszewski, Rica Gonen, Martin Kouteck{\`y}, and Nimrod Talmon.
\newblock Opinion diffusion and campaigning on society graphs.
\newblock {\em Journal of Logic and Computation}, 32(6):1162--1194, 2022.

\bibitem[\protect\citeauthoryear{Frischknecht \bgroup \em et al.\egroup }{2013}]{frischknecht2013convergence}
Silvio Frischknecht, Barbara Keller, and Roger Wattenhofer.
\newblock Convergence in (social) influence networks.
\newblock In {\em Distributed Computing: 27th International Symposium, DISC 2013, Jerusalem, Israel, October 14-18, 2013. Proceedings 27}, pages 433--446. Springer, 2013.

\bibitem[\protect\citeauthoryear{G{\"a}rtner and N.~Zehmakan}{2017}]{gartner2017color}
Bernd G{\"a}rtner and Ahad N.~Zehmakan.
\newblock Color war: Cellular automata with majority-rule.
\newblock In {\em International Conference on Language and Automata Theory and Applications}, pages 393--404. Springer, 2017.

\bibitem[\protect\citeauthoryear{G{\"a}rtner and Zehmakan}{2020}]{gartner2020threshold}
Bernd G{\"a}rtner and Ahad~N Zehmakan.
\newblock Threshold behavior of democratic opinion dynamics.
\newblock {\em Journal of Statistical Physics}, 178:1442--1466, 2020.

\bibitem[\protect\citeauthoryear{Gauy \bgroup \em et al.\egroup }{2025}]{gauy2025votermodelmeetsrumour}
Marcelo~Matheus Gauy, Anna Abramishvili, Eduardo Colli, Tiago Madeira, Frederik Mallmann-Trenn, Vinícius~Franco Vasconcelos, and David~Kohan Marzagão.
\newblock Voter model meets rumour spreading: A study of consensus protocols on graphs with agnostic nodes [extended version], 2025.

\bibitem[\protect\citeauthoryear{Hagberg \bgroup \em et al.\egroup }{2008}]{networkx}
Aric~A. Hagberg, Daniel~A. Schult, and Pieter~J. Swart.
\newblock Exploring network structure, dynamics, and function using networkx.
\newblock In Ga\"el Varoquaux, Travis Vaught, and Jarrod Millman, editors, {\em Proceedings of the 7th Python in Science Conference}, pages 11 -- 15, Pasadena, CA USA, 2008.

\bibitem[\protect\citeauthoryear{Hassin and Peleg}{1999}]{Has99}
Yehuda Hassin and David Peleg.
\newblock Distributed probabilistic polling and applications to proportionate agreement.
\newblock In Jir{\'i} Wiedermann, Peter van Emde~Boas, and Mogens Nielsen, editors, {\em Automata, Languages and Programming}, pages 402--411, Berlin, Heidelberg, 1999. Springer Berlin Heidelberg.

\bibitem[\protect\citeauthoryear{Kempe \bgroup \em et al.\egroup }{2003}]{kempe2003maximizing}
David Kempe, Jon Kleinberg, and {\'E}va Tardos.
\newblock Maximizing the spread of influence through a social network.
\newblock In {\em Proceedings of the ninth ACM SIGKDD international conference on Knowledge discovery and data mining}, pages 137--146, 2003.

\bibitem[\protect\citeauthoryear{Khuller \bgroup \em et al.\egroup }{1999}]{khuller1999budgeted}
Samir Khuller, Anna Moss, and Joseph~Seffi Naor.
\newblock The budgeted maximum coverage problem.
\newblock {\em Information processing letters}, 70(1):39--45, 1999.

\bibitem[\protect\citeauthoryear{Lesfari \bgroup \em et al.\egroup }{2022}]{lesfari2022biased}
Hicham Lesfari, Fr{\'e}d{\'e}ric Giroire, and St{\'e}phane P{\'e}rennes.
\newblock Biased majority opinion dynamics: Exploiting graph $ k $-domination.
\newblock In {\em IJCAI 2022-International Joint Conference on Artificial Intelligence}, 2022.

\bibitem[\protect\citeauthoryear{Leskovec and Krevl}{2014}]{snapnets}
Jure Leskovec and Andrej Krevl.
\newblock {SNAP Datasets}: {Stanford} large network dataset collection.
\newblock \url{http://snap.stanford.edu/data}, June 2014.

\bibitem[\protect\citeauthoryear{Li and Zehmakan}{2023}]{li2023graph}
Sining Li and Ahad~N Zehmakan.
\newblock Graph-based generalization of galam model: Convergence time and influential nodes.
\newblock {\em Physics}, 5(4):1094--1108, 2023.

\bibitem[\protect\citeauthoryear{Li \bgroup \em et al.\egroup }{2018}]{li2018influence}
Yuchen Li, Ju~Fan, Yanhao Wang, and Kian-Lee Tan.
\newblock Influence maximization on social graphs: A survey.
\newblock {\em IEEE Transactions on Knowledge and Data Engineering}, 30(10):1852--1872, 2018.

\bibitem[\protect\citeauthoryear{Lin and Lui}{2015}]{lin2015analyzing}
Yishi Lin and John~CS Lui.
\newblock Analyzing competitive influence maximization problems with partial information: An approximation algorithmic framework.
\newblock {\em Performance Evaluation}, 91:187--204, 2015.

\bibitem[\protect\citeauthoryear{Liu \bgroup \em et al.\egroup }{2023}]{liu2023fast}
Changan Liu, Xiaotian Zhou, Ahad~N Zehmakan, and Zhongzhi Zhang.
\newblock A fast algorithm for moderating critical nodes via edge removal.
\newblock {\em IEEE Transactions on Knowledge and Data Engineering}, 2023.

\bibitem[\protect\citeauthoryear{Lu \bgroup \em et al.\egroup }{2015}]{lu2015competition}
Wei Lu, Wei Chen, and Laks~VS Lakshmanan.
\newblock From competition to complementarity: Comparative influence diffusion and maximization.
\newblock {\em Proceedings of the VLDB Endowment}, 9(2), 2015.

\bibitem[\protect\citeauthoryear{Myers and Leskovec}{2012}]{myers2012clash}
Seth~A Myers and Jure Leskovec.
\newblock Clash of the contagions: Cooperation and competition in information diffusion.
\newblock In {\em 2012 IEEE 12th international conference on data mining}, pages 539--548. IEEE, 2012.

\bibitem[\protect\citeauthoryear{N~Zehmakan and Galam}{2020}]{n2020rumor}
Ahad N~Zehmakan and Serge Galam.
\newblock Rumor spreading: A trigger for proliferation or fading away.
\newblock {\em Chaos: An Interdisciplinary Journal of Nonlinear Science}, 30(7), 2020.

\bibitem[\protect\citeauthoryear{Nemhauser \bgroup \em et al.\egroup }{1978}]{nemhauser1978analysis}
George~L Nemhauser, Laurence~A Wolsey, and Marshall~L Fisher.
\newblock An analysis of approximations for maximizing submodular set functions—i.
\newblock {\em Mathematical programming}, 14:265--294, 1978.

\bibitem[\protect\citeauthoryear{Noorazar}{2020}]{noorazar2020recent}
Hossein Noorazar.
\newblock Recent advances in opinion propagation dynamics: A 2020 survey.
\newblock {\em The European Physical Journal Plus}, 135:1--20, 2020.

\bibitem[\protect\citeauthoryear{Out and Zehmakan}{2021}]{out2021majority}
Charlotte Out and Ahad~N Zehmakan.
\newblock Majority vote in social networks: Make random friends or be stubborn to overpower elites.
\newblock {\em arXiv preprint arXiv:2109.14265}, 2021.

\bibitem[\protect\citeauthoryear{Petsinis \bgroup \em et al.\egroup }{2023}]{petsinis2023maximizing}
Petros Petsinis, Andreas Pavlogiannis, and Panagiotis Karras.
\newblock Maximizing the probability of fixation in the positional voter model.
\newblock In {\em Proceedings of the AAAI Conference on Artificial Intelligence}, volume~37, pages 12269--12277, 2023.

\bibitem[\protect\citeauthoryear{Poljak and Turz{\'\i}k}{1986}]{poljak1986pre}
Svatopluk Poljak and Daniel Turz{\'\i}k.
\newblock On pre-periods of discrete influence systems.
\newblock {\em Discrete Applied Mathematics}, 13(1):33--39, 1986.

\bibitem[\protect\citeauthoryear{Shirzadi and Zehmakan}{2024}]{shirzadi2024stubborn}
Mohammad Shirzadi and Ahad~N Zehmakan.
\newblock Do stubborn users always cause more polarization and disagreement? a mathematical study.
\newblock {\em arXiv preprint arXiv:2410.22577}, 2024.

\bibitem[\protect\citeauthoryear{Wilczynski}{2019}]{wilczynski2019poll}
Ana{\"e}lle Wilczynski.
\newblock Poll-confident voters in iterative voting.
\newblock In {\em Proceedings of the AAAI Conference on Artificial Intelligence}, volume~33, pages 2205--2212, 2019.

\bibitem[\protect\citeauthoryear{Zehmakan \bgroup \em et al.\egroup }{2024}]{viral}
Ahad~N Zehmakan, Xiaotian Zhou, and Zhongzhi Zhang.
\newblock Viral marketing in social networks with competing products.
\newblock In {\em Proceedings of the 23rd International Conference on Autonomous Agents and Multiagent Systems}, pages 2047--2056, 2024.

\bibitem[\protect\citeauthoryear{Zehmakan}{2019}]{zehmakan2019spread}
Abdolahad~N Zehmakan.
\newblock {\em On the spread of information through graphs}.
\newblock PhD thesis, ETH Zurich, 2019.

\bibitem[\protect\citeauthoryear{Zehmakan}{2020}]{zehmakan2020opinion}
Ahad~N Zehmakan.
\newblock Opinion forming in erd{\H{o}}s--r{\'e}nyi random graph and expanders.
\newblock {\em Discrete Applied Mathematics}, 277:280--290, 2020.

\bibitem[\protect\citeauthoryear{Zehmakan}{2021}]{zehmakan2021majority}
Ahad~N Zehmakan.
\newblock Majority opinion diffusion in social networks: An adversarial approach.
\newblock In {\em Proceedings of the AAAI Conference on Artificial Intelligence}, volume~35, pages 5611--5619, 2021.

\bibitem[\protect\citeauthoryear{Zehmakan}{2023}]{zehmakan2023random}
Ahad~N Zehmakan.
\newblock Random majority opinion diffusion: Stabilization time, absorbing states, and influential nodes.
\newblock {\em arXiv preprint arXiv:2302.06760}, 2023.

\bibitem[\protect\citeauthoryear{Zehmakan}{2024}]{zehmakan2024majority}
Ahad~N Zehmakan.
\newblock Majority opinion diffusion: when tie-breaking rule matters.
\newblock {\em Autonomous Agents and Multi-Agent Systems}, 38(1):21, 2024.

\end{thebibliography}

\newpage~\newpage
\appendix

\section{Proof of Lemma~\ref{ineq}}
\label{ineq-appendix}

\begin{itemize}
\item \textit{Proof of inequality A}:

We have that
\begin{equation*}
\begin{split}
&\ln \left(1 - \frac{1}{l} + d\right) < \ln \left(1+d\right) < d\\ &
\Rightarrow l \ln \left(1 - \frac{1}{l} + d\right) +1 < ld+1 = \tau \\ &
\Rightarrow e^{-\frac{\tau-1}{l}} < \frac{1}{1 - \frac{1}{l} + d} \\ &
\Rightarrow \left(1 - \frac{1}{l} + d\right) e^{-\frac{\tau-1}{l}} < 1
\end{split}
\end{equation*}

Furthermore, using $(1-x)\le e^{-x}$ and then the above inequality, we get

\begin{equation*}
\begin{split}
&\left( 1- \frac{1}{l} \right)^l < e^{-1} \\ &  \Rightarrow \left( 1- \frac{1}{l} \right)^{\tau-1} < e^{-\frac{\tau-1}{l}} \\ &
\Rightarrow \left(1 - \frac{1}{l} + d\right) \left( 1- \frac{1}{l} \right)^{\tau-1} < 1.
\end{split}
\end{equation*}

Rearranging the terms, we get the desired inequality.

\item \textit{Proof of inequality B}:

We have that 
\begin{equation*}
\label{eq-b}
\begin{split}
&\log \left(2 \cdot \frac{1}{\epsilon}\right) < \log \left(e \cdot \frac{1}{\epsilon}\right) = 1 + \log \left(\frac{1}{\epsilon}\right) < \frac{1}{\epsilon} \\ & \Rightarrow \log \left(\frac{\epsilon}{2}\right) > - \frac{1}{\epsilon}
\end{split}
\end{equation*}

This implies
\begin{equation}
\label{eq-b}
\begin{split}
\frac{\epsilon}{2} > e^{-\frac{1}{\epsilon}}
\end{split}
\end{equation}

Furthermore, we calculate
\begin{equation*}
\begin{split}
&\left( 1- \frac{1}{l} \right)^l < e^{-1} \Rightarrow \left( 1- \frac{1}{l} \right)^{\tau-1} < e^{-\frac{\tau-1}{l}} \\ &= e^{-\frac{l d+1-1}{l}} = e^{-d} < e^{-\eps} < e^{-\frac{1}{\epsilon}} < \frac{\epsilon}{2}.
\end{split}
\end{equation*}
For the last step, we used Equation~\eqref{eq-b}.

Now, consider
\[
 \left(1-\frac{1}{e}+\epsilon\right) \leq 1.
 \]
Multiplying this with the previous inequality, we get 
\[
\left(1-\frac{1}{l}\right)^{\tau-1}\left(1-\frac{1}{e}+\epsilon\right) < \frac{\epsilon}{2} \cdot 1.
\]

\item \textit{Proof of inequality C}: By the previous inequality, the denominator of RHS is greater than $\frac{\epsilon}{2}$. Also $$\frac{1}{e} - \left(1-\frac{1}{l}\right)^{\tau-1} < \frac{1}{e}.$$ Thus, the RHS is less than $$\frac{2}{e\epsilon} < \frac{1}{\epsilon} \leq \eps = d < d+1.$$

\end{itemize}

\section{Proof of Lemma~\ref{a-vs-lemma}}
Since $A'\not\subset V_{\mathcal{S}}$, there exists $v \in A'$ such that $ v \notin V_{\mathcal{S}}$. Then, it suffices to show that there is a $w \in V_{\mathcal{S}}$ such that $(A' \setminus \{v\}) \cup \{w\}$ gives a higher value than the one by $A'$.

The leaves of the graph, i.e., $V_\mathcal{S}$, are all stubborn nodes, some coloured blue and some uncoloured. At any round, a node $o_j$ will turn blue if it picks a node $s_k$ that is blue. Also, since there are no red nodes, once a node turns blue it will not change colour. Thus, a node $o_j$ is blue at time $\tau$ if and only if either $o_j \in A'$ or $o_j$ picks some $s_k \in A'$ in one of the rounds. If $o_j$ turns blue at round $t$, all of nodes $o_j^1, \cdots o_j^{\eps}$ will turn blue in round $t+1$ since they only have one out-neighbour: $o_j$. Thus, a node $o_j^j$ is blue at time $\tau$ if and only if $o_j$ is blue at time $\tau -1$. A node $s_k$ is stubborn and hence blue at time $\tau$ if and only if $s_k \in A'$. Probability of $o_j$ turning blue at any round is $\frac{|N_+(o_j) \cap A'|}{|N_+(o_j)|}$. Thus, if $|N_+(o_j) \cap A'|=0$ and $o_j \notin A'$ we can see that the probability of $o_j$ being blue at time $\tau$ is 0 since it wasn't initially blue, and doesn't turn blue in any round. Then all $o_j^j$ also have probability 0 of being blue. Conversely, if $o_j \in A'$, then it is blue with probability 1. Furthermore, since $l\geq 2$ and $\eps \geq 2$, we have $\tau \geq 4$. So there is at least one round. So all $o_j^j$ are blue with probability 1. This leaves the case where $|N_+(o_j) \cap A'| \neq 0$ and $o_j \notin A'$. $o_j$ is not blue at a time $t$ if it didn't pick a blue node in the first $t$ rounds. We have $|N_+(o_j) \cap A'| \geq 1$ and $|N_+(o_j)| \leq l$. The second inequality comes because $N_+(o_j) \subset V_{\mathcal{S}}$ and $|V_\mathcal{S}| = l$. Thus, $o_j$ picks a blue node at each round with a probability $P$, where $\frac{1}{l} \leq P \leq 1$. Then, $1-\left(1-\frac{1}{l}\right)^\tau \leq S_\tau(o_j,b) \leq 1$ and $1-\left(1-\frac{1}{l}\right)^{\tau-1} \leq S_\tau(o_j^i,b) \leq 1$.

Back to $v$, there are 5 kinds of nodes $v$ can be. We will handle them as 5 cases.

\begin{itemize}
    \item \textbf{Case 1:} $v = o_j^i$ for some $i, j$ and $o_j \in A'$. Since $o_j \in A'$, $o_j^i$ is going to turn blue even if it isn't in $A'$, then $F_\tau(A') = F_\tau(A' \setminus \{v\})$. But since in $V_\mathcal{S}$ are blue with probability 0 if they are not in $A'$ and 1 if they are. Thus, any $w \in V_\mathcal{S} \setminus A'$ satisfies $F_\tau(A') + 1 \leq F_\tau(A' \cup \{w\})$ and we are done.

\item \textbf{Case 2:} $v = o_j^i$ for some $i, j$ and $|N_+(o_j) \cap A'| \neq 0$ and $o_j \notin A'$. If $v$ were to be removed from $A'$, it would still be blue with probability $S_\tau(o_j^i,b) \geq 1-\left(1-\frac{1}{l}\right)^{\tau-1}$. Thus, $F_\tau(A') - F_\tau(A' \setminus \{v\}) \leq \left(1-\frac{1}{l}\right)^{\tau-1} < 1$. So as in the previous case, any $w \in V_\mathcal{S} \setminus A'$ satisfies.

\item \textbf{Case 3:} $v = o_j^i$ for some $i, j$ and $|N_+(o_j) \cap A'| = 0$ and $o_j \notin A'$.
Here, $F_\tau(A') - F_\tau(A' \setminus \{v\}) = 1$, but since $O_j \in \bigcup_{S \in \mathcal{S}} S, |N_+(o_j)| \neq 0$. And any $w \in N_+(o_j)$, if blue, will also give a probability $\left(1-\left(1-\frac{1}{l}\right)^\tau\right)$ to $o_j$ and $\left(1-\left(1-\frac{1}{l}\right)^{\tau-1}\right)$ to all $o_j^i$. So $F_\tau((A' \setminus \{v\}) \cup \{w\}) - F_\tau(A' \setminus \{v\}) \geq 1 + 1-\left(1-\frac{1}{l}\right)^\tau+d\left(1-\left(1-\frac{1}{l}\right)^{\tau-1}\right) > 1$. This satisfies our $w$ and we are done.

\item \textbf{Case 4:} $v = o_j$ for some $j$ and $|N_+(o_j) \cap A'| \neq 0$. If $o_j$ weren't initially blue, it would still turn blue with probability at least $1-\left(1-\frac{1}{l}\right)^\tau$. So the increment is at most$\left(1-\frac{1}{l}\right)^\tau$. Also, all $o_j^i$ are now turning blue with probability 1 which will be reduced to at least $\left(1-\frac{1}{l}\right)^{\tau-1}$. So
\begin{equation*}
\begin{split}
&F_\tau(A') - F_\tau(A' \setminus \{v\}) \\ & \leq \left(1-\frac{1}{l}\right)^\tau+d\left(1-\left(1-\frac{1}{l}\right)^{\tau-1}\right)  
\end{split}
\end{equation*}

which by Lemma~\ref{ineq} (A) is less than $1$. Thus, our $w$ from cases 1 and 2 satisfies again.

\item \textbf{Case 5:} $v = o_j$ for some $j$ and $|N_+(o_j) \cap A'| = 0$. Similar to Case 3, 

\begin{equation*}
\begin{split}
F_\tau(A') - F_\tau(A' \setminus \{v\}) = 1+d
\end{split}
\end{equation*}
and 

\begin{equation*}
\begin{split}
&F_\tau((A' \setminus \{v\}) \cup \{w\}) - F_\tau(A' \setminus \{v\}) \\ & \geq 1 + 1-\left(1-\frac{1}{l}\right)^\tau+d\left(1-\left(1-\frac{1}{l}\right)^{\tau-1}\right).
\end{split}
\end{equation*}

Then 
\begin{equation*}
\begin{split}
&F_\tau((A' \setminus \{v\}) \cup \{w\}) - F_\tau(A') \\ & \geq 1 -\left(1-\frac{1}{l}\right)^\tau - d\left(\left(1-\frac{1}{l}\right)^{\tau-1}\right)
\end{split}
\end{equation*}
and again by Lemma~\ref{ineq} (A), we are done.
\end{itemize}

\section{Proof of Lemma~\ref{lemma4}}
\label{lemma4-appendix}
We start by the left inequality. As we explained before, we have $$am \leq k + mc + d \ mc.$$ Rearranging gives us $$\frac{am - k}{1+d} \leq mc$$ which is equivalent to
\[
\frac{(am-k)\left(1-\left(1-\frac{1}{l}\right)^{\tau-1}\right)}{(1+d)\left(1-\left(1-\frac{1}{l}\right)^{\tau-1}\right)} \leq mc
\]
Furthermore, we have

\[
\left(1-\frac{1}{l}\right)^{\tau} < \left(1-\frac{1}{l}\right)^{\tau-1}.
\]

Combining the above two inequalities, we get
\begin{equation*}
\begin{split}
&\frac{(am-k)\left(1-\left(1-\frac{1}{l}\right)^{\tau-1}\right)}{1-\left(1-\frac{1}{l}\right)^{\tau}+d\left(1-\left(1-\frac{1}{l}\right)^{\tau-1}\right)} \\ &< \frac{(am-k)\left(1-\left(1-\frac{1}{l}\right)^{\tau-1}\right)}{(1+d)\left(1-\left(1-\frac{1}{l}\right)^{\tau-1}\right)}
\\ & \Rightarrow \frac{(am-k)\left(1-\left(1-\frac{1}{l}\right)^{\tau-1}\right)}{1-\left(1-\frac{1}{l}\right)^{\tau}+d\left(1-\left(1-\frac{1}{l}\right)^{\tau-1}\right)} < mc
\end{split}
\end{equation*}
Now for the right inequality, as we discussed before we have

\begin{equation*}
    \begin{split}
        &k + mc \left(1-\left(1-\frac{1}{l}\right)^\tau \right) \\ &+ d \ mc \left(1-\left(1-\frac{1}{l}\right)^{\tau-1}\right) \leq am
    \end{split}
\end{equation*}

and rearranging directly gives us the required result.

\section{Proof of Lemma~\ref{noopt}}
\label{noopt-appendix}
Let $MC$ be solution obtained from $\{S_j:s_j\in A\}$. For the sake of contradiction, suppose that there is a solution $MC' > MC$. Let $AM'$ be the corresponding solution to the Adoption Maximisation problem achieved from the solution $MC'$. We show that $AM' > AM$ contradicting optimality of $AM$.

We know that $MC\leq m$. Thus, we get

\[
1 + MC < e^{MC} \leq e^m.
\]

and 
\[
\left(1 - \frac{1}{l}\right)^{\tau - 1} < e^{-\frac{\tau-1}{l}} = e^{-d} < e^{-m}.
\]

So the product of the above inequalities yields 

\[
(MC+1)\left(1 - \frac{1}{l}\right)^{\tau - 1} < 1
\]

Also, 
\[
\left(1-\frac{1}{l}\right)^{\tau} < \left(1-\frac{1}{l}\right)^{\tau-1}<1.
\]

As we explained before, we know that
\begin{equation*}
\begin{split}
AM' &\geq k + MC' \left(1-\left(1-\frac{1}{l}\right)^\tau \right) \\ &+ d \ MC' \left(1-\left(1-\frac{1}{l}\right)^{\tau-1}\right)
\end{split}
\end{equation*}

Furthermore, since the solutions are integer values $MC' > MC$ implies $$MC' \geq MC + 1.$$

Now, combining the above inequalities we can conclude
\begin{equation*}
\begin{split}
AM' &\geq k + (MC+1) \left(1-\left(1-\frac{1}{l}\right)^\tau \right) \\ &+ d \ (MC+1) \left(1-\left(1-\frac{1}{l}\right)^{\tau-1}\right) \\ &
=k + MC + d \ MC + 1 - (MC+1)\left(1-\frac{1}{l}\right)^{\tau} \\ & + d \left( 1 - (MC+1) \left(1-\frac{1}{l}\right)^{\tau-1}\right) \\ & >k + MC + d \ MC \\ & \geq AM
\end{split}
\end{equation*}
Thus, $AM' > AM$ and we are done.

\section{Calculations from Theorem~\ref{hardness}}
\label{cal-hardness-appendix}

Calculations for Equation~\eqref{eq-star} are as follows. We first use the left inequality from Lemma~\ref{lemma4} on the solution given by our algorithm, which gives us
\begin{equation*}
\overline{MC} > \frac{(\overline{AM}-k)\left(1-\left(1-\frac{1}{l}\right)^{\tau-1}\right)}{D}.
\end{equation*}

Now, applying that $\overline{AM}$ is at least an $(1-\frac{1}{e}+\epsilon)$ approximation of $AM$, we get
\begin{equation*}
\begin{split}
\overline{MC} & > \frac{\left(\left(1 - \frac{1}{e}+\epsilon\right)AM-k\right)\left(1-\left(1-\frac{1}{l}\right)^{\tau-1}\right)}{D} \\ & = \left(\frac{\left(1 - \frac{1}{e}+\epsilon\right)(AM-k)}{D} + \frac{\left(-\frac{1}{e}+\epsilon\right)k}{D}\right)\\ &\times \left(1-\left(1-\frac{1}{l}\right)^{\tau-1}\right).
\end{split}
\end{equation*}

By Lemma~\ref{noopt}, the solution to Maximum Coverage corresponding to the optimal solution AM is the optimal solution MC. Now, applying the right inequality of Lemma~\ref{lemma4} and opening the brackets:

\begin{equation*}
\begin{split}
\overline{MC} & >\left(1-\frac{1}{e}\right)MC + \epsilon MC \\ &- \left(1 - \frac{1}{e} + \epsilon \right)MC\left(1-\frac{1}{l}\right)^{\tau-1} + \\ & \frac{\left(\epsilon-\frac{1}{e}\right)k}{D}\left(1-\left(1-\frac{1}{l}\right)^{\tau-1}\right) \hfill
\end{split}
\end{equation*}

Calculations from Equation~\eqref{eq-cal-2} using Lemma~\ref{ineq} are as follows:
Start by adding inequality $A$ and $C$
\begin{equation*}
    \begin{split}
        & d + 1 + 1-\left(1-\frac{1}{l}\right)^\tau - d\left(1-\frac{1}{l}\right)^{\tau-1} \\ 
        & > \frac{\frac{1}{e} - \left(1-\frac{1}{l}\right)^{\tau-1}}{\epsilon - \left(1-\frac{1}{l}\right)^{\tau-1}\left(1-\frac{1}{e}+\epsilon\right)} \\
        & \Rightarrow D+1 > \frac{\frac{1}{e} - \left(1-\frac{1}{l}\right)^{\tau-1}}{\epsilon - \left(1-\frac{1}{l}\right)^{\tau-1}\left(1-\frac{1}{e}+\epsilon\right)} \\
        & \Rightarrow D > \frac{\frac{1}{e} - \left(1-\frac{1}{l}\right)^{\tau-1}}{\epsilon - \left(1-\frac{1}{l}\right)^{\tau-1}\left(1-\frac{1}{e}+\epsilon\right)} -1 \\
        & \Rightarrow D > \frac{\frac{1}{e} - \left(1-\frac{1}{l}\right)^{\tau-1} - \epsilon + \left(1-\frac{1}{l}\right)^{\tau-1}\left(1-\frac{1}{e}+\epsilon\right)}{\epsilon - \left(1-\frac{1}{l}\right)^{\tau-1}\left(1-\frac{1}{e}+\epsilon\right)} \\ 
& \Rightarrow \epsilon - \left(1-\frac{1}{l}\right)^{\tau-1}\left(1-\frac{1}{e}+\epsilon\right) \\ &> \frac{\frac{1}{e} - \left(1-\frac{1}{l}\right)^{\tau-1} - \epsilon + \left(1-\frac{1}{l}\right)^{\tau-1}\left(1-\frac{1}{e}+\epsilon\right)}{D} \\ & = \frac{\left(\frac{1}{e} - \epsilon\right)\left(1-\left(1-\frac{1}{l}\right)^{\tau-1}\right)}{D}   
\end{split}
\end{equation*}

Calculations to show the RHS of Equation~\eqref{eq-star} is larger than $(1-1/e)MC$ are given below. Since the LHS of Equation~\eqref{eq-cal-2} is positive, and $MC \geq k$, we have:
\begin{equation*}
    \begin{split}
        & MC \left( \epsilon - \left(1-\frac{1}{l}\right)^{\tau-1}\left(1-\frac{1}{e}+\epsilon\right) \right) \\
        & > \frac{\left(\frac{1}{e} - \epsilon\right)k\left(1-\left(1-\frac{1}{l}\right)^{\tau-1}\right)}{D} \\
        & \Rightarrow 0 > - MC \left( \epsilon - \left(1-\frac{1}{l}\right)^{\tau-1}\left(1-\frac{1}{e}+\epsilon\right) \right) \\ 
        &- \frac{\left(\epsilon - \frac{1}{e}\right)k\left(1-\left(1-\frac{1}{l}\right)^{\tau-1}\right)}{D}
    \end{split}
\end{equation*}
Adding this last inequality to Equation ~\eqref{eq-star} and opening brackets, we get
$$\overline{MC} > \left(1-\frac{1}{e}\right)MC.$$
\section{Proof of Lemma~\ref{nodseq}}
\label{nodseq-appendix}
We will use induction. For $t=0$, we see that this is true because the only node in the node sequence of $v$ is itself so the time 0 colour of $v$ is the time 0 colour of itself. Let the statement be true for $t=i$. Then, for $t=i+1$, a node is uncoloured if and only if it is uncoloured at time $t$ and the node it picks is also uncoloured at time $t$. By the induction hypothesis, this is true if and only if all the nodes in $\gamma_{t}(v)$ and $\gamma_{t}(PS_\tau(v,t))$ were initially uncoloured. Then, by concatenation, all the nodes in $\gamma_{t+1}(v)$ were also initially uncoloured. If $v$ was coloured at time $t$, and it picks an uncoloured node, it should retain its colour. All the nodes in $\gamma_{t}(PS_\tau(v,t))$ were initially uncoloured, and $\gamma_t(v)$ has an initially coloured node. Then, the first coloured node in the concatenation is the first coloured node in $\gamma_t(v)$, retaining its colour. If $v$ picks a coloured node, it adopts its colour. $\gamma_t(PS_\tau(v,t))$ has an initially coloured node. Then, the first coloured node in the concatenation is the first coloured node in $\gamma_t(PS_\tau(v,t))$, adopting its colour. Thus, the induction holds and the statement is true for all $0 \leq t \leq \tau$.

\section{Proof of Lemma~\ref{strcon}}
\label{appendix:convlemma}

We will prove this by showing that there is a path from any non-monochromatic colouring to a monochromatic one. Then, any strongly connected component that does not contain a monochromatic colouring can not be absorbing, but the monochromatic colourings have out degree 0, thus are singleton strongly connected components that are also absorbing. Thus, these are the only absorbing strongly connected components of $C_G$.

Consider an arbitrary non-monochromatic colouring $S$ of $G$. If there are no blue nodes in $S$, then there have to be only red and uncoloured nodes. Consider the update sequence, where every node which has a red out-neighbour picks it. Any node that doesn't have a red out-neighbour will pick an uncoloured neighbour and retain its colour. Then, since there is at least one red node, after $t$ rounds, every node at a distance at most $t$ from it will be red and in a number of rounds equal to the diameter of $G$, all nodes will be red. Therefore, in this case, there is a path from $S$ to a monochromatic colouring in the corresponding Markov chain.

If there are blue nodes in $S$, consider a blue node $v$ and an update sequence where every node that has a blue out-neighbour picks it. Since $G$ is aperiodic, there exist cycles $C_1$ to $C_l$ the HCF of whose lengths $n_1$ to $n_l$ is 1. There also exist paths $P_i$ of length $m_i$ from $C_{i+1}$ to $C_i$ for each $1 \leq i \leq l-1$ and a $P_0$ of length $m_0$ from $C_1$ to $v$. These paths can be trivial if the cycles share a point.
For each $C_i$, let $P_{i-1} \cap C_i = a_i$ and $P_{i} \cap C_i = b_i$. Then, in $m_0$ rounds, $a_1$ is blue, after which it is blue after every $n_1$ round, since consecutive nodes on the cycle $C_1$ will be blue after each round. Since $b_1$ is on the cycle, it will turn blue somewhere between rounds $m_1$ and $m_1+n_1$ after which it will turn blue after every $n_1$ rounds. Each time $b_1$ turns blue, $a_2$ will turn blue after $m_1$ rounds due to $P_1$. Thus, after the first time $a_2$ turns blue, it will turn blue after every $n_1$ rounds since $b_1$ turns blue after every $n_1$ rounds. Additionally, due to the cycle $C_2$, each time $a_2$ turns blue, it will turn blue again after $n_2$ rounds. Thus, it will turn blue again after $\lambda n_1 + \mu n_2$ for each $\lambda , \mu \in \mathbb{N}$. By the extended Euler algorithm, after some rounds, $a_2$ will turn blue after every $HCF(n_1,n_2)$ rounds. Through $C_2$, $b_2$ will also turn blue with a period of $HCF(n_1,n_2)$. Now for induction let $b_i$ turn blue with period $HCF(n_1, n_2, \cdots n_i)$. Then, after $m_i$ rounds, $a_{i+1}$ turns blue after every $HCF(n_1, n_2, \cdots n_i)$ rounds. Through $C_{i+}$, $a_{i+1}$ turns blue after $n_{i+1}$ rounds each time it turns blue. So by extended Euclid's algorithm $a_{i+1}$, and through $C_{i+1}$ $b_{i+1}$, turns blue after every $HCF(HCF(n_1, n_2, \cdots n_i), n_{i+1}) = HCF(n_1, n_2, \cdots n_i, n_{i+1})$ rounds. Through induction, we see that $b_l$ turns blue with a period of $HCF(n_1,n_2, \cdots n_l) = 1$. Then, in $diam(G)$ rounds, the whole graph turns blue.

\section{Proof of Lemma~\ref{eqrel}}
\label{eqrel-appendix}

\begin{itemize}
    \item \textbf{Reflexivity.} It is trivial that $u \sim u$ since $\mu(u,u) = 0$.
    \item \textbf{Symmetry.} We show that if $u \sim v$, then $\gamma | \mu(u,v)$. Consider the shortest path from $u$ to $v$ followed by the shortest path from $v$ to $u$. This is a walk starting and ending at the same point $u$ so its length must be a multiple of $\gamma$. Since both $\mu(u,v)+\mu(v,u)$ and $\mu(u,v)$ are multiples of $\gamma$, so is $\mu(v,u)$ and thus $v \sim u$.
    \item \textbf{Transitivity}. Let $u \sim v$ and $v \sim w$. Then, by appending the shortest paths from $u$ to $v$ and $v$ to $w$, we have a $u$-$w$ walk with length multiple of $\gamma$. Appending the shortest path from $w$ to $u$, we get a walk that starts and ends at $u$ and so must have length divisible by $\gamma$. Thus, the length of the shortest path from $w$ to $u$ must have length divisible by $\gamma$ which means $w \sim u$. (By symmetry, $u \sim w$.)
\end{itemize}

We see that the relation given above is an equivalence relation due to the fact that all cycles have length to be multiples of $\gamma$. Additionally, since $\gamma$ is not 1, there are no edges between nodes in the same class as that would be a path of length 1. So each period class is an independent set.

\section{Proof of Lemma~\ref{unicycle}}
\label{unicycle-appnedix}
Consider any $v \in V$ and define the relation $\sim_v$ as $x \sim_v y \Leftrightarrow \gamma | \mu(v,x) - \mu(v,y)$. This partitions $V$ into $\gamma$ classes. We know none of the $\gamma$ classes are empty because there is a cycle of length at least $\gamma$ so there are nodes at distance $1$ to $\gamma -1$ from $v$. Furthermore, for any $x$ with $\mu(v,x) \equiv i \ \mod \ \gamma$ and $(x,y) \in E$, appending the edge to any walk from $v$ to $x$ gives a walk from $v$ to $y$ of length $\equiv i+1 \ \mod \gamma$. So the equivalence classes satisfy the cyclic property that there are only edges from one class to the next, giving the required structure on the contracted graph. Next, we show that $\sim_v = \sim$. Consider any $x \sim y$. Let $P$ be the shortest path from $v$ to $x$ and $Q$ from $x$ to $y$. Let the length of $P \equiv i \ mod \ \gamma$. Then on joining $P$ and $Q$, we get a path of length $\equiv i \ mod \ \gamma$ from $v$ to $y$. Thus the length of shortest path from $v$ to $y$ is also $\equiv i \ mod \ \gamma \Rightarrow x \sim_v y$. Now consider any $x \sim_v y$. Let $P$ be the shortest path from $v$ to $x$, $Q$ from $v$ to $y$, and $R$ from $x$ to $y$. Let length of $P \equiv i \ mod \ \gamma$. Then $Q \equiv i \ mod \ \gamma$. For any path $S$ from $y$ to $v$, both $PRS$ and $QS$ are walks starting and ending at $v$, so have length divisible by $\gamma$. $|P| + |R| + |S| \equiv |Q| + |S| \ mod \ \gamma \Rightarrow |R| \equiv 0 \ mod \ \gamma \Rightarrow x \sim y$. Thus, $\sim_v$ and $\sim$ are equivalent. Since they are equal, and the equivalence classes of $\sim_v$ satisfy the required property, so do the classes of $\sim$.

\section{Proof of Lemma~\ref{strap}}
\label{strap-appendix}
For any $u,v \in \Gamma$, consider a $u$-$v$ path in $G$. The length of this path must be a multiple of $\gamma$, and every $\gamma$-th node in the path must be in $\Gamma$. Having a path of length exactly $\gamma$ between 2 consecutive such nodes implies there's an edge between them in $G_\Gamma$. Replacing each $\gamma$ edges of the $u$-$v$ path with the respective edge gives us a $u$-$v$ path in $G_\Gamma$. Since $u$ and $v$ were arbitrary, $G_\Gamma$ is strongly connected.

Now, suppose $G_\Gamma$ is periodic with period $\rho$. Then, every cycle has length multiple of $\rho$. Since $\gamma\rho$ is not the period of $G$, there is a cycle whose length is not a multiple of $\gamma\rho$. Starting from a node in $\Gamma$, if we list all edges in this cycle and then contract every $\gamma$ edges into the corresponding edge in $G_\Gamma$, we get a cycle whose length is not a multiple of $\rho$. A contradiction. Thus, $G_\Gamma$ is aperiodic.

\section{Further Experimental Results on Algorithms}
\label{appendix:exp}
Figure below includes the experimental findings for all networks, including the ones omitted in the main text.

\begin{figure}[h]
    \includegraphics[width=0.45\linewidth]{fb0.png}
    \includegraphics[width=0.45\linewidth]{fb414.png} \\
    \includegraphics[width=0.45\linewidth]{otc.png}
    \includegraphics[width=0.45\linewidth]{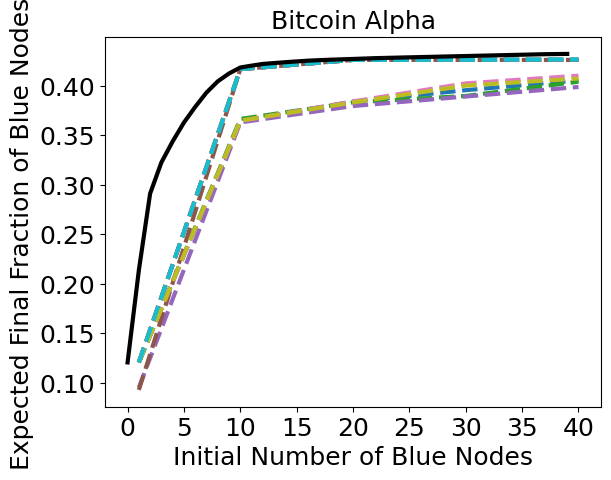} \\
    \includegraphics[width=0.45\linewidth]{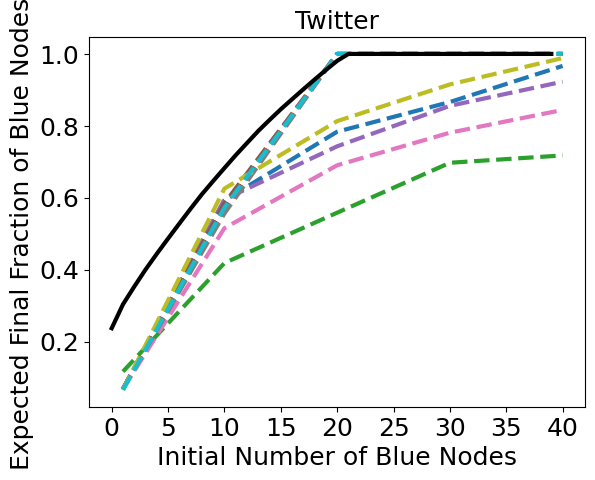}
    \includegraphics[width=0.45\linewidth]{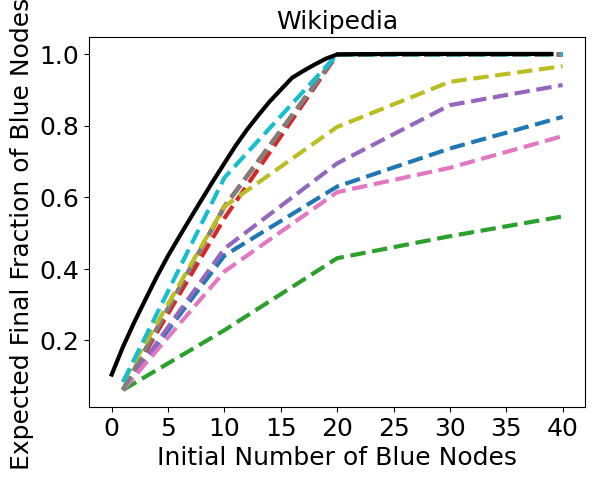} \\
    \includegraphics[width=0.45\linewidth]{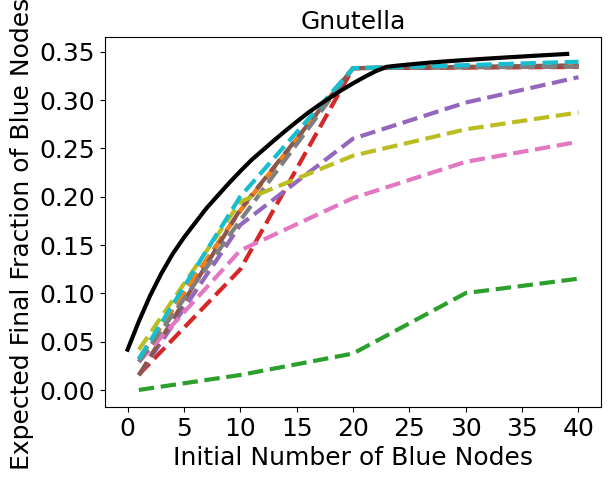} \hspace{0.15\linewidth}
    \includegraphics[width=0.3\linewidth]{legend.png} 
    
    \caption{Performance of Greedy algorithm against some well known centrality measures. In each graph, the final expected fraction of blue nodes is plotted against the budget. Each run has 20 red nodes and a budget varying from 1 to 40 for $\tau=20$ rounds.}
    \label{fig:graphs-appendix}
\end{figure}

\section{Experiments on Convergence Time}
\label{appendix:convergence-exp}

To complement our theoretical results on convergence time, we generate Hyperbolic Random Graph (HRG) of increasing size and simulate our model on them from different initial colourings. The graphs are generated using the algorithm provided in~\cite{hgg}. In the algorithm, the value of model parameter $R$ was taken to be 5 and $\alpha$ to be 1.0 since these values gave suitable edge to node ratio matching the real world networks. Additionally, since we are testing the bounds for strongly connected graphs, we check if the graph is strongly connected. If it is not, we discard the graph and generate a new one. To colour the graph, we adopt four different methods. 
\begin{enumerate}
\item We colour two diametrically opposite points with red and blue, and the rest uncoloured.
\item Each node is assigned a random colour, with probability of either red or blue being 1\% each and uncoloured being 98\%.
\item Each node is assigned a random colour, red or blue with 50\% probability. There are no uncoloured nodes.
\item All the blue nodes are chosen as the $\frac{n}{2}$ nodes closest to a fixed node $v$. The remaining nodes are red, there are no uncoloured nodes.
\end{enumerate}

The above methods aim to mimic scenarios where the process takes the longest. Thus, they would give an estimate of the maximum convergence time.
The idea for approach 1 and 4 is to minimise the initial interaction between red and blue nodes so that they take maximum time to overcome each other. Furthermore, the first two experiments have most of the nodes initially uncoloured and the last two experiments have all nodes initially coloured. The algorithms for finding diameter, shortest distances, and testing strongly connectedness come from the Networkx library~\cite{networkx}.

Since we do not have an algorithm to directly calculate the expected convergence time (this problem is suspected to be \#P-hard), we use a Monte Carlo approach, generating a HRG, running the simulation, and noting the convergence time. We then average the convergence time over 128 random graphs for each size.

The average number of rounds against graph size is plotted in Figure~\ref{fig:conv} for a number of initial colourings. We see that colouring all the nodes takes longer on average to converge than leaving them uncoloured. This might mean that the process converges through a path where all coloured nodes spread to uncoloured nodes in a way other than by random neighbours or nearest out-neighbours so initialising it this way takes longer to converge. Also, in both cases, the random colouring takes longer on average. So having random neighbours of a different colour can take longer to converge than having some structure to the initial state.

\begin{figure}[t]
\centering
\includegraphics[width=0.4\textwidth]{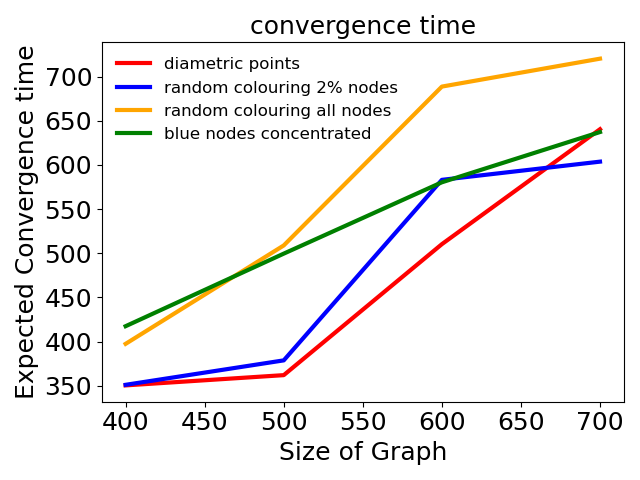}
\caption{Average convergence time for graphs of each size with four different strategies to initially colour the nodes.}
\label{fig:conv}
\end{figure}

We observe that the computed values are significantly smaller than our theoretical bound, which grows in a cubic rate in the size of the graph. Thus, it raises the question whether tighter bounds can be derived for special classes of graphs, such as HRG, which are more likely to be present in the real-world.

\end{document}